\documentclass[a4paper,UKenglish,cleveref,thm-restate]{lipics-v2021}

\pdfoutput=1

\hideLIPIcs

\title{The Power of Counting Steps in Quantitative Games}

\author{Sougata Bose}{University of Liverpool, UK}{sougata.bose@liverpool.ac.uk}{https://orcid.org/0000-0003-3662-3915}{}
\author{Rasmus Ibsen-Jensen}{University of Liverpool, UK}{}{https://orcid.org/0000-0003-4783-0389}{}
\author{David Purser}{University of Liverpool, UK}{D.Purser@liverpool.ac.uk}{https://orcid.org/0000-0003-0394-1634}{}
\author{Patrick Totzke}{University of Liverpool, UK}{totzke@liverpool.ac.uk}{https://orcid.org/0000-0001-5274-8190}{}
\author{Pierre Vandenhove}{LaBRI, Université de Bordeaux, France}{pierre.vandenhove@u-bordeaux.fr}{https://orcid.org/0000-0001-5834-1068}{}

\authorrunning{S. Bose, R. Ibsen-Jensen, D. Purser, P. Totzke and P. Vandenhove}

\Copyright{Sougata Bose, Rasmus Ibsen-Jensen, David Purser, Patrick Totzke and Pierre Vandenhove}



\category{}

\relatedversion{Extended version of the conference paper~\cite{BIPTV24}.}

\funding{\emph{P.~Vandenhove} is funded by ANR project \emph{G4S} (ANR-21-CE48-0010-01).
\emph{S.~Bose} and \emph{P.~Totzke} acknowledge support from the EPSRC (EP/V025848/1, EP/X042596/1). Collaboration initiated via Royal Society International Exchanges grant (IES{\textbackslash}R1{\textbackslash}201211).}

\nolinenumbers

\usepackage{tikz}
\usetikzlibrary{arrows,shapes,automata,calc,decorations.pathreplacing}
\usepackage{tikz}
\usepackage{pgf}
\usepackage{pgfplots}

\usetikzlibrary{arrows,decorations,intersections,pgfplots.fillbetween}
\usetikzlibrary{shapes.multipart,shapes.misc}
\usepgflibrary{decorations.pathreplacing}
\usetikzlibrary{arrows,topaths}
\usetikzlibrary{automata,shapes.arrows, backgrounds, fadings,intersections, positioning}
\usetikzlibrary{shadows}

\usetikzlibrary{decorations.pathmorphing}
\usetikzlibrary{decorations.shapes}
\usetikzlibrary{shapes}
\usetikzlibrary{backgrounds}
\usetikzlibrary{fit}
\usetikzlibrary{arrows}
\usetikzlibrary{calc}

\usetikzlibrary{quotes}

\pgfplotsset{compat=1.14}%

\tikzset{every picture/.style={
thick
}}
\tikzset{every state/.style={initial text={}}}
\tikzset{Max/.style={state,draw,circle,minimum size=9*1.5,inner sep=0}}
\tikzset{Min/.style={state,draw,rectangle,minimum size=9*1.5,inner sep=0}}
\tikzset{Imp/.style={draw,minimum size=5,inner sep=0}}

\tikzset{label/.style={font=\scriptsize,auto}}
\tikzset{every edge/.style={draw,->}}

\definecolor{myorange}{RGB}{255, 128, 0}

\tikzset{
	wrap/.style={
		line cap=round,
		#1,
		line width=21pt,
		opacity=0.3,
	},}
\tikzset{squish/.style={decorate,decoration={snake,amplitude=.5mm,segment length=4mm,post length=0.9mm,pre length=0mm}}}

\usepackage{pgfplots}

\usepackage{adjustbox}

\newcommand{\IN}{\mathbb{N}}
\newcommand{\IZ}{\mathbb{Z}}
\newcommand{\IQ}{\mathbb{Q}}
\newcommand{\IR}{\mathbb{R}}
\newcommand{\card}[1]{\lvert{#1}\rvert}

\renewcommand{\epsilon}{\varepsilon}

\newcommand{\clr}{c}
\newcommand{\colours}{C}
\newcommand{\word}{w}
\newcommand{\arena}{\mathcal{A}}
\newcommand{\vertex}{v}
\newcommand{\vertices}{V}
\newcommand{\edge}{e}
\newcommand{\edges}{E}
\newcommand{\from}{\mathsf{from}}
\newcommand{\too}{\mathsf{to}}
\newcommand{\col}{\mathsf{col}}
\newcommand{\arenaFull}{(\vertices, \vertices_1, \vertices_2, \edges)}
\newcommand{\emptyHistory}{\lambda}
\newcommand{\history}{h}
\newcommand{\histories}{\mathsf{hists}}
\newcommand{\play}{\rho}
\newcommand{\product}{\otimes}
\newcommand{\SC}{\mathcal{S}}

\newcommand{\objective}{O}

\newcommand{\strat}{\sigma}
\newcommand{\memory}{\mathcal{M}}
\newcommand{\memStates}{M}
\newcommand{\memState}{m}
\newcommand{\memInit}{\memState_0}
\newcommand{\memUpd}{\delta}
\newcommand{\memUpdExt}{\delta^*}
\newcommand{\memoryFull}{(\memStates, \memInit, \memUpd)}
\newcommand{\step}{s}
\newcommand{\reward}{r}

\newcommand{\MP}{\mathsf{MP}}
\newcommand{\TP}{\mathsf{TP}}
\newcommand{\MPsupge}{\overline{\MP}_{\ge 0}}
\newcommand{\MPsupg}{\overline{\MP}_{> 0}}
\newcommand{\MPinfge}{\underline{\MP}_{\ge 0}}
\newcommand{\MPinfg}{\underline{\MP}_{> 0}}
\newcommand{\TPsupge}{\overline{\TP}_{\ge 0}}
\newcommand{\TPsupinf}{\overline{\TP}_{= +\infty}}
\newcommand{\TPinfinf}{\underline{\TP}_{= +\infty}}

\newcommand{\TPinfminusinf}{\underline{\TP}_{> -\infty}}
\newcommand{\TPsupg}{\overline{\TP}_{> 0}}
\newcommand{\TPinfge}{\underline{\TP}_{\ge 0}}
\newcommand{\TPinfg}{\underline{\TP}_{> 0}}

\newcommand{\thresh}{r}

\newcommand{\MPinfgr}{\underline{\MP}_{> \thresh}}

\newcommand{\SigmaI}[1]{\mathbf{\Sigma}_{#1}^0}
\newcommand{\PiI}[1]{\mathbf{\Pi}_{#1}^0}
\newcommand{\SigmaOne}{\SigmaI{1}}
\newcommand{\PiOne}{\PiI{1}}
\newcommand{\SigmaTwo}{\SigmaI{2}}
\newcommand{\PiTwo}{\PiI{2}}
\newcommand{\SigmaThree}{\SigmaI{3}}

\newcommand{\prefixOrd}{\preceq}
\newcommand{\strictPrefixOrd}{\prec}

\newcommand{\winning}{W_{\arena,1}}
\newcommand{\winningExtended}{\winning'}

\tikzstyle{circ}=[draw,circle,minimum height=7mm]
\tikzstyle{square}=[draw,minimum width=6.5mm,minimum height=6.5mm]

\newcommand{\+}[1]{\mathbb{#1}}

\newcommand{\N}{\+{N}}

\newcommand{\x}{\times}
\newcommand{\eps}{\varepsilon}

\usepackage{xspace}
\newcommand{\playerI}{Player~$p$\xspace}
\newcommand{\POne}{Player~{1}\xspace}
\newcommand{\PTwo}{Player~{2}\xspace}

\newcommand{\hops}{\Step{*}{}{}}

\newcommand{\Step}[3]{\ensuremath{\,{\stackrel{#1}{\longrightarrow}}{}^{\scriptstyle{#2}}_{\scriptstyle{#3}}}\,}

\newcommand{\inEdge}[4]{$\begin{tikzpicture}[baseline]%
\node[anchor=base] (A) {#1};%
\node[anchor=base] (B) [right of=A, node distance=4em] {#3};%
\draw[->] (A) edge[#4] node[midway,above=-2pt,font=\scriptsize] {#2} (B);%
\end{tikzpicture}$}

\usepackage{graphicx, floatrow}
\usepackage{makecell}

\begin{document}

\maketitle

\begin{abstract}
We study deterministic games of infinite duration played on graphs and focus on the strategy complexity of quantitative objectives.
Such games are known to admit optimal memoryless strategies over finite graphs, but require infinite-memory strategies in general over infinite graphs.

We provide new lower and upper bounds for the strategy complexity of \emph{mean-payoff} and \emph{total-payoff} objectives over infinite graphs, focusing on whether \emph{step-counter strategies} (sometimes called \emph{Markov strategies}) suffice to implement winning strategies.
In particular, we show that over finitely branching arenas, three variants of $\limsup$ mean-payoff and total-payoff objectives admit winning strategies that are based either on a step counter or on a step counter and an additional bit of memory.
Conversely, we show that for certain $\liminf$ total-payoff objectives, strategies resorting to a step counter and finite memory are not sufficient.
For step-counter strategies, this settles the case of all classical quantitative objectives up to the second level of the Borel hierarchy.

\end{abstract}

\section{Introduction} \label{sec:intro}
Two-player (zero-sum, turn-based, perfect-information) games on graphs are an established formalism in formal verification, especially for \emph{reactive synthesis}~\cite{BCJ18,fijalkow2023games}.
They are used to model the interaction between a system, trying to satisfy a given \emph{specification}, against an uncontrollable environment, assumed to act antagonistically as a worst case.
We can model the system and its environment as two opposing players, called \emph{\POne} and \emph{\PTwo} respectively, who move a token through the graph of possible system configurations (called the \emph{arena}).
The specification is modelled as a
winning condition (called \emph{objective} henceforth), which is a set of all those interactions that the system player deems acceptable.
The main algorithmic task when using this approach for formal verification is \emph{solving} such games: given an arena, an objective, and an initial vertex, decide whether the system player has a \emph{winning strategy}, which corresponds to a controller for the system that guarantees that the specification holds no matter the behaviour of the environment.
Additionally, reactive synthesis aims to \emph{synthesise} (compute a representation of) a winning strategy if one exists.%

\subparagraph*{Strategy complexity}
To synthesise winning strategies, it is useful to know what kind of resources ``suffice'', i.e., are needed to implement a winning strategy, should one exist.
This naturally depends on the model used for the interaction (the size and topology of the arena) and on the specification (the type of objective and whether probabilistic or absolute guarantees are required).
We assume that strategies make decisions based on some internal memory, that stores and updates an abstraction of the past play.

The simplest strategies are those that are \emph{memoryless}, meaning they base their decisions solely on the current arena vertex.
Games on finite arenas where
memoryless strategies are sufficient to win
can usually be solved in $\mathsf{NP} \cap \mathsf{coNP}$~\cite{Pur95}
and winning strategies effectively synthesised.
This is true for \emph{parity}, \emph{discounted-payoff}~\cite{Sha53}, \emph{mean-payoff}~\cite{EM79}, and \emph{total-payoff}~\cite{CdHS03,GZ04} objectives.
Even beyond finite graphs, memoryless strategies may suffice in more general contexts, such as for parity objectives over arenas of arbitrary cardinality~\cite{EJ91,Zie98}, or \emph{discounted-payoff objectives} over finitely branching arenas~\cite[Corollary~2.1]{OhlmannThesis}.\footnote{
Thus we consider the strategy complexity in discounted-payoff games as settled for the setting we consider.
On infinitely branching arenas, step-counter strategies are insufficient (see \cref{fig:LB-0a}).
}
For concurrent (stochastic) \emph{reachability} games on finite arenas, memoryless strategies also suffice~\cite{BBL2022,KMST2024}.

Generally more powerful than memoryless strategies are \emph{finite-memory} strategies, which refer to strategies that can be implemented with a finite-state (Mealy) machine.
A canonical class of languages over infinite words,
and standard for defining objectives in games, are
the $\omega$-regular languages~\cite{Rab69,GTW02}.
One of the celebrated related results about reactive synthesis is the \emph{finite-memory determinacy} of $\omega$-regular games~\cite{BL69,Rab69,GH82}, which means that if there is a winning strategy in a game on a finite arena and with an $\omega$-regular objective, there is one that can be implemented with a simple finite-state machine (whose size can be bounded).
This implies that games with $\omega$-regular objectives can be solved and that strategies can be synthesised, since it bounds the search space for winning strategies.
Remarkably, the existence of winning finite-memory strategies for $\omega$-regular games even holds over arbitrary infinite arenas~\cite{Zie98}.
When finite-memory strategies are sufficient, one of the main questions is usually to \emph{minimise} their size,
i.e., to find winning strategies with as few memory states as possible~\cite{DJW97,Cas22,CFH24,BFRV23,BIT2024}.%

\begin{figure}
    \centering
    \begin{subfigure}[b]{0.35\textwidth}
    \begin{tikzpicture}[node distance=1.5cm and 1.5cm]
\draw (0,0.75) node {};
\draw (0,-0.75) node {};
        \draw (0,0) node[Min,initial] (av1) {$s$};
        \draw ($(av1)+(1.5,0)$) node[Max] (av2) {$t$};
        \draw ($(av2)+(1.5,0)$) node[Max] (av3) {$q$};
        \draw (av1) edge[-latex'] node[above=2pt] {$-\N$} (av2);
        \draw (av2) edge[-latex'] node[above=2pt] {$\N$} (av3);
        \draw (av3) edge[-latex', loop right] node[above=2pt] {$0$} (av3);
\end{tikzpicture}
    \caption{The arena $\arena_1$.}
    \label{fig:LB-0a}
    \end{subfigure}
    \hspace{0.1\textwidth}
    \begin{subfigure}[b]{0.35\textwidth}
    \begin{tikzpicture}[yshift=1cm,node distance=1.5cm and 1.5cm]
\draw (0,0.75) node {};
\draw (0,-0.75) node {};
        \draw (0,0) node[Min,initial] (v1) {$s$};
        \draw ($(v1)+(1.5,0)$) node[Max] (v2) {$t$};
        \draw (v1) edge[-latex',out=30,in=150] node[above=2pt] {$-\N$} (v2);
        \draw (v2) edge[-latex',out=-150,in=-30] node[below=2pt] {$\N$} (v1);
\end{tikzpicture}
    \caption{The arena $\arena_1'$. }
    \label{fig:LB-0b}
    \end{subfigure}
  \caption{Arenas implementing the ``match the number'' game. Circles designate vertices controlled by \POne  and squares designate \PTwo. The edge labels indicate that for every $i\in\N$ there is a distinct edge with weight $-i$ from $s$ to $t$, and $+i$ from $t$ to $q$ or from $t$ to $s$. For $\arena_1$, consider the objective ``sum of weights exceeds $0$''. \POne can always match and thus win, but needs unbounded memory.
  The arena $\arena'_1$ shows a repeated version
  for the $\limsup$ \emph{mean}-payoff objective.
  }
\label{fig:LB-0}
\end{figure}

Already very simple games require infinite memory to win.
This especially holds for quantitative objectives,
which ask that the aggregate of individual edge weights along a play exceeds some threshold.
For instance, consider
a game where the environment picks a number and then the controller has to pick a larger one (see \cref{fig:LB-0a}).
In order to win, \POne has to remember the (per se unbounded) initial challenge and no finite memory structure would be sufficient to do so.
This objective is not $\omega$-regular since it is built upon an infinite alphabet.
We seek to understand
for different classes of games, what kind of infinite-memory structures are sufficient for winning strategies.

A natural, arguably the simplest, type of infinite memory structure is a \emph{step counter}: it only remembers how many steps have elapsed since the start of the game.
The availability of such a counter is a reasonable assumption for practical applications,
as most embedded devices have access to the current time, which suffices when each step takes a fixed amount of time.
A \emph{step-counter strategy} is one that, in addition to the current arena vertex, has access to the number of steps elapsed.
Notice that in the game in \cref{fig:LB-0a}, a step counter does not provide any relevant information (every path to vertex $t$ has length one). Therefore, step-counter strategies do not suffice for \POne.
An important ingredient for these counterexamples is that the underlying arena is infinitely branching (and uses arbitrary weights).
For many classes of games on \emph{finitely} branching arenas, strategies based on a step counter and additional finite memory are close to being the simplest kinds of strategies sufficient to win.
Examples are especially prevalent in stochastic games. For instance, in the ``Big Match'' (a concurrent mean-payoff game on a finite arena),
neither a step counter nor finite memory is sufficient to play $\eps$-optimally, yet a step counter \emph{together with} one bit is~\cite{HIN2018}. The same is true for the ``Bad Match'', which can be presented as a B\"uchi (repeated reachability) game~\cite{MaitraSudderth:DiscreteGambling,thuijsman1992optimality, KMST2024.2}. This upper bound holds generally for concurrent B\"uchi games on finite arenas~\cite{KMST2024.2}.

\subparagraph*{Quantitative objectives}
Objectives based on numerical weights are commonly called \emph{quantitative objectives}. These are defined using \emph{quantitative payoff functions}, which combine any finite sequence of weights into an aggregate number.
The three more common ones are the discounted-payoff~\cite{Sha53}, mean-payoff~\cite{Gil57,EM79}, and total-payoff functions~\cite{FV96,CdHS03}.
Every payoff function induces four variants of objectives, depending on whether we consider the $\limsup$ or $\liminf$, and on whether we ask that the limit is larger or strictly larger than a threshold.
For total payoff, it is also relevant to distinguish the use of real values or $\infty$ as a threshold.
We give an example to describe informally how we denote such objectives: $\MPsupge$ refers to the set of infinite sequences of rational numbers that achieve a value $\ge 0$ for the $\limsup$ variant (the line is above $\MP$) of the mean-payoff function (specified by letters $\MP$).
Over infinite arenas, the four variants are not equivalent and infinite-memory strategies are needed for at least one of the players (see~\cite[Example~8.10.2]{Put94} and~\cite{OS24}).

To study the strategy complexity for different quantitative objectives, we classify them according to which level of the \emph{Borel hierarchy} they belong to (which also ensures that the games we consider are determined~\cite{Martin}).
In the first level of the hierarchy lie the \emph{open} and \emph{closed} objectives (i.e., the sets respectively in $\SigmaOne$ and $\PiOne$), for which there exist recent characterisations of the sufficient memory structures
over finite or infinite arenas~\cite{CFH24,BFRV23}.
We build on this to establish upper bounds for more complex objectives.
All variants of mean-payoff and total-payoff objectives are on the second or third level of the Borel hierarchy.
Ohlmann and Skrzypczak~\cite{OS24} study objectives through their topological properties and provide a characterisation of the \emph{prefix-independent} $\SigmaTwo$ objectives for which memoryless strategies suffice for \POne over arbitrary arenas.
They show in particular that memoryless strategies suffice for \POne for the quantitative objectives $\MPinfg$ and $\TPinfminusinf$, even over infinitely branching arenas.
Over stochastic games, quantitative (in particular $\liminf$ mean-payoff) objectives on infinite arenas
generally do not have ($\eps$-)optimal strategies based on a step counter, even for finitely branching Markov decision processes~\cite{MM23}.

\subparagraph*{Our contributions}
We settle the strategy complexity over infinite, deterministic games
for the mean-payoff and total-payoff objectives up to the second level of the Borel hierarchy.
In particular, we show for which of these, step-counter strategies are sufficient for \POne.
Our upper bounds all allow for arenas with arbitrary weights, while our strongest lower bounds only use weights $-1$, $0$, and $1$.
Our results are as follows and summarised in \cref{fig:table}.
\begin{itemize}
    \item For $\TPinfg$ and $\TPinfge$, strategies using a step counter and an arbitrary amount of finite memory do not suffice, even over acyclic finitely branching arenas (\cref{thm:insufficientTPinf}, \cref{sec:lowerBounds}).
    The proof rules out finite-memory structures using an application of the \emph{infinite Ramsey theorem} to allow \PTwo to stay winning in a particular infinite arena regardless of the finite-memory structure of \POne.
    \item In \cref{sec:PiTwo}, we provide a sufficient condition for when step-counter strategies suffice over finitely branching arenas for prefix-independent objectives in $\PiTwo$, i.e, countable intersections of open sub-objectives (\cref{thm:Pi2StepCounter}). This implies in particular that step-counter strategies do suffice for $\MPsupge$ and $\TPsupinf$ (\cref{cor:MPTPsup}), which is tight in the sense that finite-memory strategies do not suffice for these objectives,
    even over acyclic finitely branching arenas (\cref{lem:FMnotSufficient}).
    The proof uses carefully constructed expanding ``bubbles'', so that within each consecutive bubble, \POne can satisfy the next open sub-objective. The step counter is used to determine the current bubble.
    \item In \cref{sec:NonPrefIndPiTwo}, we show that for $\TPsupge$, which is not prefix-independent, strategies using a step-counter and one additional bit of memory suffice (\cref{thm:TPsupge}).
    This is tight in
    that neither finite-memory strategies nor step-counter strategies suffice, even over acyclic finitely branching arenas (\cref{lem:FMnotSufficient,lem:noSCforTP}). The proof similarly employs bubbles, but an additional bit is needed to keep track of whether a ``sub-objective'' has been achieved in the current bubble and then switches to stay in the winning region.
\end{itemize}

\subparagraph*{Structure}
We define the various notions used throughout the paper in \cref{sec:prelim}.
\cref{sec:lowerBounds} is dedicated to all lower bounds on the strategy complexity of the various objectives, culminating in a lower bound for $\TPinfge$.
\cref{sec:openClosed} is devoted to recalling useful results on open and closed objectives, upon which the following sections build.
\cref{sec:PiTwo} proves a sufficient condition for the sufficiency of step-counter strategies for prefix-independent $\PiTwo$ objectives.
\cref{sec:NonPrefIndPiTwo} proves an upper bound on the strategy complexity of $\TPsupge$.

This article is the extended version of conference paper~\cite{BIPTV24}, including in particular all missing proofs in its appendix.

\begin{table}
\centering
\bgroup
\def\arraystretch{1.2}

\begin{tabular}{l|l|c|p{4.79cm}}
    Obj. & Description & Class & Strategy complexity \\\hline\hline

    $\MPinfg$ & $\bigcup_{m\ge 1} \bigcup_{i \ge 1} \bigcap_{j \ge i} \{\word \mid \MP(\word_{\le j}) \ge \frac{1}{m} \}$ & $\SigmaTwo$ & \multirow{2}{4.9cm}{Memoryless (even over infinitely branching arenas)~\cite{OS24}} \\\cline{1-3}

    $\TPinfminusinf$ & $\bigcup_{m\ge 1} \bigcup_{i \ge 1} \bigcap_{j \ge i} \{\word \mid \TP(\word_{\le j}) \ge -m \}$ & $\SigmaTwo$ &
    \\\hline

    $\TPinfg$ & $\bigcup_{m\ge 1} \bigcup_{i \ge 1} \bigcap_{j \ge i} \{\word \mid \TP(\word_{\le j}) \ge \frac{1}{m} \}$ & $\SigmaTwo$ & SC\,+\,FM insufficient (\cref{thm:insufficientTPinf}) \\\hline

    $\MPsupge$ & $\bigcap_{m\ge 1} \bigcap_{i \ge 1} \bigcup_{j \ge i} \{\word \mid \MP(\word_{\le j}) \ge \frac{-1}{m} \}$ & $\PiTwo$ &\multirow{2}{4.9cm}{SC sufficient (\cref{cor:MPTPsup}) \\
    FM insufficient (\cref{lem:FMnotSufficient})}
    \\\cline{1-3}

    $\TPsupinf$ & $\bigcap_{m\ge 1} \bigcap_{i \ge 1} \bigcup_{j \ge i} \{\word \mid \TP(\word_{\le j}) \ge m \}$ & $\PiTwo$ &
    \\\hline

    $\TPsupge$ & $\bigcap_{m\ge 1} \bigcap_{i \ge 1} \bigcup_{j \ge i} \{\word \mid \TP(\word_{\le j}) \ge \frac{-1}{m} \}$ & $\PiTwo$ &\makecell[cl]{SC\,+\,1-bit sufficient (\cref{thm:TPsupge}) \\FM insufficient (\cref{lem:FMnotSufficient}) \\
    SC insufficient (\cref{lem:noSCforTP})\vspace{-0.25cm}
    }
\end{tabular}
\egroup
\caption{Results for quantitative objectives up to the second level of the Borel hierarchy for finitely branching arenas. \emph{SC} refers to \emph{step counter}, and \emph{FM} refers to \emph{finite memory}.
}
\label{fig:table}
\end{table}

\section{Preliminaries} \label{sec:prelim}
Given a set $X$, we write $X^*$ for the set of finite words on $X$, $X^+$ for the set of non-empty finite words on $X$, and $X^\omega$ for the set of infinite words on $X$.
For $\word\in X^*$, we write $\card{\word}$ for the length of $\word$.
For $\word\in X^\omega$ and $j\in\IN$, we write $\word_{\le j}$ for the finite prefix of length $j$ of $\word$.%

\subparagraph*{Games}
We study two-player zero-sum \emph{games}, each given by an \emph{arena} and an \emph{objective}, as defined below.
We refer to the two opposing players as \POne and \PTwo.

An \emph{arena} is a directed graph with two kinds of vertices where edges are labelled by an element of $\colours$, a non-empty set of \emph{colours}.
Formally, an arena is a tuple $\arena = \arenaFull$ where $\vertices = \vertices_1 \cup \vertices_2$ is a non-empty set of \emph{vertices}, $\vertices_1$ and $\vertices_2$ are disjoint, and $\edges \subseteq \vertices \times \colours \times \vertices$ is a set of labelled \emph{edges}.
Vertices in $\vertices_1$ and $\vertices_2$ are respectively controlled by \POne and \PTwo, which will appear clearly when we define strategies below.
We require that for every vertex $\vertex \in \vertices$, there is an edge $(\vertex, \clr, \vertex')\in\edges$ (arenas are ``non-blocking'').
For $\edge = (\vertex, \clr, \vertex')$, we write $\from(\edge)$ for $\vertex$, $\col(\edge)$ for~$\clr$, and $\too(\edge)$ for $\vertex'$.
An arena is \emph{finite} if $\vertices$ is finite, and \emph{finitely branching} if for every $\vertex\in\vertices$, the set $\{\edge\in\edges\mid \from(\edge) = \vertex\}$ is finite.

A \emph{history} is a finite sequence $\history = \edge_1\ldots\edge_n\in\edges^*$ of edges such that for $i\in\{1, \ldots, n-1\}$, $\too(\edge_i) = \from(\edge_{i+1})$.
We write $\from(\history)$ for $\from(\edge_1)$, $\too(\history)$ for $\too(\edge_n)$, and $\col(\history)$ for the sequence $\col(\edge_1)\ldots\col(\edge_n)\in\colours^*$.
For convenience, we assume that for every vertex $\vertex$, there is a distinct \emph{empty history} $\emptyHistory_\vertex$ such that $\from(\emptyHistory_\vertex) = \too(\emptyHistory_\vertex) = \vertex$.
The set of histories of $\arena$ is denoted as $\histories(\arena)$.
For $p\in\{1, 2\}$, we write $\histories_p(\arena)$ for the set of histories $\history$ such that $\too(\history)\in\vertices_p$.
A \emph{play} is an infinite sequence of edges $\play = \edge_1\edge_2\ldots\in\edges^\omega$ such that for $i \ge 1$, $\too(\edge_i) = \from(\edge_{i+1})$.
We write $\from(\play)$ for $\from(\edge_1)$ and $\col(\play)$ for $\col(\edge_1)\col(\edge_2)\ldots\in\colours^\omega$.
A~history $\history$ (resp.\ a play $\play$) is said to be \emph{from $\vertex$} if $\vertex = \from(\history)$ (resp.\ $\vertex = \from(\play)$).

An \emph{objective} (sometimes called a \emph{winning condition} in the literature) is a set $\objective \subseteq \colours^\omega$.
An objective $\objective$ is \emph{prefix-independent} if for all $\word\in\colours^*$, $\word'\in\colours^\omega$, $\word\word'\in\objective$ if and only if $\word'\in\objective$.%

\subparagraph*{Strategies}
A \emph{strategy of \playerI on $\arena$} is a function $\strat\colon \histories_p(\arena) \to \edges$ such that for all~$\history\in\histories_p(\arena)$, $\from(\strat(\history)) = \too(\history)$.
A play $\play = \edge_1\edge_2\ldots$ is \emph{consistent with a strategy $\strat$ of \playerI} if for all finite prefixes $\history$ of $\play$ such that $\too(\history) \in \vertices_p$, $\strat(\history) = \edge_{\card{\history} + 1}$.
A strategy $\strat$ of \POne is \emph{winning for objective $\objective$ from a vertex~$\vertex$} if all plays from $\vertex$ consistent with $\strat$ induce a sequence of colours in $\objective$.
For a fixed objective, the set of vertices of an arena $\arena$ from which a winning strategy for \POne exists is called the \emph{winning region of \POne on $\arena$} and is denoted $\winning$.
A strategy $\strat$ of \POne is \emph{uniformly winning for objective $\objective$ in $\arena$} if $\strat$ is winning from every vertex of the winning region of $\arena$.

A \emph{memory structure for an arena $\arena = \arenaFull$} is a tuple $\memory = \memoryFull$ where~$\memStates$ is a set of \emph{memory states}, $\memInit\in\memStates$ is an \emph{initial state}, and $\memUpd\colon \memStates \times \edges \to \memStates$ is a \emph{memory update function}.
We extend $\memUpd$ to a function $\memUpdExt\colon \memStates \times \edges^* \to \memStates$ in a natural way.
A memory structure~$\memory$ is \emph{finite} if $\memStates$ is finite.
A strategy $\strat$ of \playerI on $\arena$ is \emph{based on $\memory$} if there exists a function $f\colon \vertices_p \times \memStates \to \edges$ such that, for all $\history\in\histories_p(\arena)$, $\strat(\history) = f(\too(\history), \memUpdExt(\memInit, \history))$.
We will abusively assume that a strategy based on a memory structure is this function $f$.

A \emph{memoryless strategy} is a strategy based on a memory structure with a single memory state.
A \emph{$1$-bit strategy} is a strategy based on a memory structure with two memory states.
A \emph{step counter} is a memory structure $\SC = (\IN, 0, (\step, \edge) \mapsto \step + 1)$ that simply counts the number of steps already elapsed in a game.
A strategy $\strat$ of \playerI on~$\arena$ is a \emph{step-counter strategy} if $\strat$ is based on a step counter; in other words, if there is a function $f\colon \vertices_p \times \IN \to \edges$ such that $\strat(\history) = f(\too(\history), \card{\history})$.
This means that $\strat$ only considers the current vertex and the number of steps elapsed to make its decisions.
Step-counter strategies are sometimes called ``Markov strategies''~\cite{thuijsman1992optimality,KMST20}.

A \emph{step-counter and finite-memory structure} is a memory structure with state space $\memStates = \IN \times \{0, \ldots, K-1\}$, initial state $(0, 0)$, and a transition function $\memUpd$ such that $\memUpd((\step, \memState), \edge) = (\step + 1, \memUpd'((\step, \memState), \edge))$ for some function $\memUpd'\colon \memStates \times \edges \to \{0, \ldots, K-1\}$.
Notice that a step counter corresponds to the special case of a step-counter and finite-memory structure with $K=1$.
A \emph{step-counter + $1$-bit strategy} is a strategy based on a step-counter and finite-memory structure with~$K=2$.

We say that a kind of strategies \emph{suffices for objective $\objective$ over a class of arenas} if, for all arenas in this class, from all vertices of her winning region, \POne has a winning strategy of this kind.
We say that a kind of strategies \emph{suffices uniformly for objective $\objective$ over a class of arenas} if, for all arenas in this class, \POne has a uniformly winning strategy of this kind.%

For an arena $\arena = \arenaFull$ and a memory structure $\memory=\memoryFull$, we write $\arena \product \memory$ for the \emph{product between $\arena$ and $\memory$}. It is the arena $(\vertices', \vertices_1', \vertices_2', \edges')$ such that $\vertices' = \vertices \times \memStates$, $\vertices_1' = \vertices_1 \times \memStates$, $\vertices_2' = \vertices_2 \times \memStates$, and $\edges' = \{((\vertex, \memState), \clr, (\vertex', \memUpd(\memState, \edge))) \mid \edge = (\vertex, \clr, \vertex')\in\edges, \memState\in\memStates\}$.
Observe that \POne has a winning strategy based on $\memory$ from a vertex $\vertex$ in an arena $\arena$ if and only if \POne has a winning memoryless strategy from vertex $(\vertex, \memInit)$ in $\arena\product\memory$.

To simplify reasonings over specific arenas, we show that step counters do not have any use when the arena already \emph{encodes the step count}.

\begin{lemma} \label{lem:encodingTheStep}
    Let $\arena = \arenaFull$ be an arena, and $\vertex_0\in\vertices$ be an initial vertex.
    Assume that for each pair of histories $\history_1, \history_2$ from $\vertex_0$ to some $\vertex\in\vertices$, we have $\card{\history_1} = \card{\history_2}$ (i.e., the arena already ``encodes the step count from $\vertex_0$'').
    Then, a step-counter and finite-memory strategy with $K$ states of finite memory can be simulated from $\vertex_0$ by a strategy with only $K$ states of finite memory.
\end{lemma}
\begin{proof}
    By hypothesis on $\arena$, there exists $n_\vertex\in\IN$ the length of any history from $\vertex_0$ to $\vertex$.
    Let $\strat'\colon\vertices_1 \times \IN \times \memStates \to \edges$ be a step-counter and finite-memory strategy with $\memStates = \{0, \ldots, K-1\}$, with finite-memory update function $\memUpd'\colon \memStates \times \edges \to \{0, \ldots, K-1\}$.
    Let $\memory = (\memStates, 0, \memUpd)$ be the memory structure with $\memUpd(\memState, \edge) = \memUpd'((n_{\from(\edge)}, \memState), \edge)$.
    By construction, the strategy $\strat\colon \vertices_1 \times \memStates \to \edges$ such that $\strat(\vertex, \memState) = \strat'(\vertex, n_\vertex, \memState)$ behaves exactly like $\strat'$ from $v_0$.
\end{proof}

\subparagraph*{Quantitative objectives}
We consider classical quantitative objectives: mean-payoff and total-payoff objectives, as defined below.
Let $\colours \subseteq \IQ$ (when colours are rational numbers, we often refer to them as \emph{weights}).
For a finite word $\word = \clr_1\ldots\clr_{\card{\word}} \in\colours^*$, define $\TP(\word) = \sum_{i = 1}^{\card{\word}} \clr_i$ for the \emph{total payoff} of the word, i.e., the sum of the weights it contains.
Further, when $\card{\word} \ge 1$, let $\MP(\word) = {\TP(\word)}/{\card{\word}}$
denote the \emph{mean payoff} of the word $w$, i.e., the mean of the weights it contains.
We extend any such aggregate function $X\colon\colours^*\to \IR$ to infinite words by taking limits: for $\word\in\colours^\omega$, we define $\overline{X}(\word) = \limsup_{j} X(\word_{\le j})$ and $\underline{X}(\word) = \liminf_{j} X(\word_{\le j})$.
Fixing a
binary relation $\triangleright\subseteq\IR^2$ and threshold $\thresh\in\IQ\cup\{-\infty,\infty\}$, this naturally defines objectives
$
    \overline{X}_{\triangleright \thresh} = \{\word\in\colours^\omega \mid \overline{X}(\word) \triangleright \thresh\}
$ and
$
    \underline{X}_{\triangleright \thresh} = \{\word\in\colours^\omega \mid \underline{X}(\word) \triangleright \thresh\}
$.

In particular, we are interested in the
limit infimum/supremum objectives for total and mean payoff.\footnote{We only consider objectives where the threshold is a \emph{lower} bound ($\triangleright \in \{>,\ge\}$); each variant with \emph{upper} bound behaves like a variant with lower bound when we replace each weight $\clr$ in arenas with its additive inverse $-\clr$ and switch the sup/inf (for instance, $\overline{\mathsf{MP}}_{<\thresh}$ behaves like $\MPinfgr$ when we invert the weights).}
We consider the mean-payoff variants with threshold $r\in\IQ$, and the total-payoff variants with $r\in\IQ\cup\{-\infty,+\infty\}$.
Note that all four mean-payoff objectives and all four total-payoff objectives with $\infty$ threshold are prefix-independent, but the four total-payoff objectives with threshold in $\IQ$ are not prefix-independent.

\begin{remark} \label{rem:threshold0}
    Our results are generally stated for threshold $r = 0$.
    This is without loss of generality since the results deal with large classes of arenas, and little modifications to the arenas allow to reduce from an arbitrary rational threshold to threshold $0$.
    \lipicsEnd
\end{remark}

\subparagraph*{Topology of objectives}
For $\word\in\colours^*$, we write $\word\colours^\omega = \{\word\word' \mid \word'\in\colours^\omega\}$ for the objective containing all infinite words that start with $\word$ (it is sometimes called the \emph{cylinder} or \emph{cone} of $\word$).
An objective $\objective$ is \emph{open} if there is a set $A \subseteq \colours^*$ such that $\objective = \bigcup_{\word\in A} \word\colours^\omega$.
For an open objective $\objective$, we say that a finite word $\word\in\colours^*$ \emph{already satisfies $\objective$} if $\word\colours^\omega \subseteq \objective$.
If an objective is open, then by definition, any infinite word it contains has a finite prefix that already satisfies it.
An objective is \emph{closed} if it is the complement of an open set.

Open and closed objectives are at the first level of the \emph{Borel hierarchy}; the set of open (resp.\ closed) objectives is denoted $\SigmaOne$ (resp.\ $\PiOne$).
For $i>1$, we can define $\SigmaI{i}$ as all the countable unions of sets in $\PiI{i-1}$, and $\PiI{i}$ as all the countable intersections of sets in $\SigmaI{i-1}$.
All the objectives considered in this paper lie in the first three levels of this hierarchy, and we focus on those on the second level.

\section{Lower bounds} \label{sec:lowerBounds}
We provide lower bounds on the size/structure of the memory to build winning strategies, focusing on objectives $\MPsupge$, $\TPsupinf$, $\TPsupge$, and $\TPinfg$, which are the four objectives on the second level of Borel hierarchy for which we want to establish whether step-counters strategies suffice.
We mention where our constructions directly work for further objectives.

All lower bounds are based on the simple idea that one player chooses some number and the other must match it.
We first observe that on infinitely branching arenas
with arbitrary weights, neither finite memory nor a step counter, nor both together, is sufficient.
The proof uses the arenas from \cref{fig:LB-0}, discussed informally in \cref{sec:intro} (full proofs in \cref{app:lowerBounds}).%

\begin{restatable}{lemma}{infBranchingLB} \label{lem:infBranchingLB}
    Over infinitely branching arenas with arbitrary weights, step-counter and finite-memory strategies
    are not sufficient for \POne for objectives
    $\MPsupg$, $\MPsupge$, $\TPsupinf$, $\TPinfg$, $\TPinfge$, $\TPsupg$ and $\TPsupge$.
\end{restatable}

We now establish lower bounds over finitely branching arenas.
Firstly, the example $\arena_1'$ can be made finitely branching and acyclic, as depicted in \cref{fig:LB-1}.
The resulting arena, $\arena_2$, simply unfolds $\arena_1'$ so that any edge $(s,-j,t)$ is replaced by a finite path $s^i_0\to \cdots \to s^i_j \to t^i_0$, and similarly for the responses.
This construction works as long as one can discourage (i.e., make losing) the choice to stay on the infinite intermediate chain of vertices and not moving on to a vertex controlled by the opponent.
Here, this is achieved by using weights $1$ on the chains of \PTwo and weights $-1$ on the chains of \POne, which are then compensated by weights twice as large.
In practice, edges with weights $i\in\IN$ (resp.\ $-i\in-\IN$) can be replaced by chains of $i$ weights $1$ (resp.\ $i$ weights $-1$).
This allows to obtain lower bounds on the $\limsup$ objectives.
The fact that finite-memory strategies are insufficient for variants of the mean-payoff objectives over finitely branching arenas was already discussed in~\cite[Example~8.10.2]{Put94} and~\cite{OS24}; we rephrase it here for completeness.

\begin{figure}
\floatsetup{valign=t, heightadjust=all}

\begin{floatrow}
\ffigbox{
\adjustbox{width=0.99\linewidth}{
        \begin{tikzpicture}[node distance=1cm and 1.5cm, on grid]

\node [Min,initial,initial where=left] (s0) at (0,0) {$s^i_0$};
\node [Min,right=of s0] (s1) {$s^i_1$};
\node [Min,right=of s1] (s2) {$s^i_2$};
\node [Min,right=1.9cm of s2] (s3) {$s^i_j$};
\node [right=1cm of s3] (s34) {$\cdots$};

\node [Min, Imp, below=of s0] (s0') {};
\node [Min, Imp, below=of s1] (s1') {};
\node [Min, Imp, below=of s2] (s2') {};
\node [Min, Imp, below=of s3] (s3') {};
\node [below=of s34] (s34') {$\cdots$};

\node [Max, below=1cm of s0'] (t0) {$t^i_0$};
\node [Max,right=of t0] (t1) {$t^i_1$};
\node [Max,right=of t1] (t2) {$t^i_2$};
\node [Max,right=1.9cm of t2] (t3) {$t^i_j$};
\node [right=1cm of t3] (t34) {$\cdots$};

\node [Max, Imp, below=of t0] (t0') {};
\node [Max, Imp, below=of t1] (t1') {};
\node [Max, Imp, below=of t2] (t2') {};
\node [Max, Imp, below=of t3] (t3') {};
\node [below=of t34] (t34') {$\cdots$};

\node [Min,below=1cm of t0'] (r0) {$s^{i+1}_0$};

\draw (s0) edge node[label] {1} (s1);
\draw (s1) edge node[label] {1} (s2);
\draw (s2) edge[dotted] node[label] {1} (s3);

\draw (s0) edge node[label,auto] {$0$} (s0');
\draw (s1) edge node[label,auto] {$-2$} (s1');
\draw (s2) edge node[label,auto] {$-4$} (s2');
\draw (s3) edge node[label,auto] {$-2j$} (s3');
\draw (s3') edge[dotted] node[label,auto,swap] {$0$} (s2');
\draw (s2') edge node[label,auto,swap] {$0$} (s1');
\draw (s1') edge node[label,auto,swap] {$0$} (s0');

\draw (s0') edge node[label,auto,swap] {$0$} (t0);

\draw (t0) edge node[label] {$-1$} (t1);
\draw (t1) edge node[label] {$-1$} (t2);
\draw (t2) edge[dotted] node[label] {$-1$} (t3);

\draw (t0) edge node[label,auto] {$0$} (t0');
\draw (t1) edge node[label,auto] {$2$} (t1');
\draw (t2) edge node[label,auto] {$4$} (t2');
\draw (t3) edge node[label,auto] {$2j$} (t3');
\draw (t3') edge[dotted] node[label,auto,swap] {$0$} (t2');
\draw (t2') edge node[label,auto,swap] {$0$} (t1');
\draw (t1') edge node[label,auto,swap] {$0$} (t0');

\draw (t0') edge node[label,auto,swap] {$0$} (r0);

\end{tikzpicture}
    }
}{\caption{The arena $\arena_2$ is acyclic and every vertex has finite in- and out-degree.
We recall that circles are controlled by \POne and squares by \PTwo.
\label{fig:LB-1}
}
 }
\ffigbox{
\adjustbox{width=0.99\linewidth}{
        \begin{tikzpicture}[node distance=2cm and 1.5cm, on grid]

\node [Min,initial,initial where=left] (s0) at (0,0) {$s_0$};
\node [Min,right=of s0] (s1) {$s_1$};
\node [Min,right=of s1] (s2) {$s_2$};
\node [Min,right=1.9cm of s2] (s3) {$s_i$};
\node [right=1cm of s3] (s34) {$\cdots$};

\node [Max, below=of s0] (t0) {$t_0$};
\node [Max,right=of t0] (t1) {$t_1$};
\node [Max,right=of t1] (t2) {$t_2$};
\node [Max,right=1.9cm of t2] (t3) {$t_i$};
\node [right=1cm of t3] (t34) {$\cdots$};

\node [Max, below=2cm of t0] (r0) {$r_0$};
\coordinate[below=1cm of t0] (joinhere);

\draw (s0) edge node[label] {$1$} (s1);
\draw (s1) edge node[label] {$1$} (s2);
\draw (s2) edge[dotted] node[label] {$1$} (s3);

\def\dist{-1}
\draw (s0) edge[squish] node[label] {$-1$} (t0);
\draw (s1) edge[squish] node[label] {$-3$} (t1); 
\draw (s2) edge[squish] node[label] {$-5$} (t2); 
\draw (s3) edge[squish] node[label] {$-2i-1$} (t3);

\draw (t0)  edge[squish] node[label] {$0$} (t1);
\draw (t1)  edge[squish] node[label] {$0$} (t2);
\draw (t2) edge[squish,dotted] node[label] {$0$} (t3);

\draw (t0) to node[label] {$0$} +(0,\dist) edge (r0);
\draw (t1) to node[label] {$1$} +(0,\dist) ++(0,\dist) -> (joinhere); 
 \draw (t2) to node[label] {$2$} +(0,\dist) ++(0,\dist) -> (joinhere); 
\draw[dotted] (t3) to node[label] {$i$} +(0,\dist) ++(0,\dist) -> (joinhere); 

\draw (r0) edge[loop right] node[label] {$0$} (r0);
\end{tikzpicture}
    }
}{\caption{The arena $\arena_3$.
        Arrows
        \inEdge{$s_i$}{$-2i-1$}{$t_{i}$}{squish} are shorthand for paths of length $2i+1$ with edge weights $-1$, and \inEdge{$t_i$}{$0$}{$t_{i+1}$}{squish} are shorthand for paths of length $3$ with edge weights $0$.
        \label{fig:LB-2}
        }
    }
\end{floatrow}
\end{figure}

\begin{lemma} \label{lem:FMnotSufficient}
    Over finitely branching arenas,
    finite-memory strategies are not sufficient for \POne for objectives
    $\MPsupg$, $\MPsupge$, $\TPsupinf$, $\TPsupg$, and $\TPsupge$.
\end{lemma}

Notice that although finite memory is insufficient for \POne in $\arena_2$, a step counter allows her to deduce an upper bound on the previous choice of \PTwo and is therefore sufficient.
Indeed, since $\arena_2$ is finitely branching and every round starts in a unique initial vertex for that round, \POne can (over) estimate that all steps of the history so far were spent by her opponent's choice (steps between $s^i_0$ up to some $s^i_j$ and then leading directly to $t^{i+1}_0$).

In order to construct an arena in which no step-counter strategy is sufficient, we obfuscate possible histories leading to \POne's choices by making them the same length (see \cref{fig:LB-2}).

\begin{lemma} \label{lem:noSCforTP}
    Consider the arena $\arena_3$ depicted in \cref{fig:LB-2}.
    \POne has a winning strategy, but no winning step-counter strategy for objectives $\TPinfg$, $\TPinfge$, $\TPsupg$, and $\TPsupge$.
    Hence, over finitely branching arenas, step-counter strategies are not sufficient for \POne for objectives $\TPinfg$, $\TPinfge$, $\TPsupg$, and $\TPsupge$.
\end{lemma}
\begin{proof}
    \POne only makes relevant choices at vertices $t_i$, and the choice is whether to \emph{delay} (move to $t_{i+1}$) or \emph{exit} (move to $r_0$).
    A winning (finite-memory) strategy for all mentioned objectives is to delay twice and then exit.
    Indeed, any history leading to $t_i$ has total payoff of at least $-i-1$. By delaying twice and then exiting, \POne guarantees that the sink vertex $r_0$ is reached and the total payoff collected on the way is at least $1$.

    Conversely, any strategy $\sigma$ of \POne that is based solely on a step counter cannot distinguish histories
    leading to the same vertex $t_i$.
    Let us assume that $\sigma$ does not choose to avoid $r_0$ indefinitely, as doing so would result in a negative total payoff, which is losing for her. Then there is at least one vertex $t_i$ from which the strategy exits. \PTwo can exploit this by going there via $s_i$.
    The resulting play has a negative total payoff.
\end{proof}

We now extend the previous examples to show that even access to both a step counter and finite memory is not sufficient for \POne. The construction below is stated for the total-payoff objective $\TPinfge$, and also works for $\TPinfg$.
The main idea is to require \POne to delay going to $r_0$ more than a constant number of times, as dictated by \PTwo's initial move.

\begin{figure}%
\centering

\floatsetup{valign=t, heightadjust=all}

\begin{floatrow}
\ffigbox{
\adjustbox{width=0.99\linewidth}{
        \begin{tikzpicture}[node distance=1.5cm and 1.5cm, on grid]

\node [Min,initial,initial where=left] (s0) at (0,0) {$s_0$};
\node [Min,right=of s0] (s1) {$s_1$};
\node [Min,right=of s1] (s2) {$s_2$};
\node [Min,right=1.9cm of s2] (s3) {$s_i$};
\node [right=1cm of s3] (s34) {$\cdots$};

\node [Max, below=of s0] (t0) {$t_0$};
\node [below right=.6cm and .9cm of t0] (t0dots) {$\cdots$};
\node [Max,right=of t0] (t1) {$t_1$};
\node [below right=.6cm and .9cm of t1] (t1dots) {$\cdots$};
\node [Max,right=of t1] (t2) {$t_2$};
\node [below right=.6cm and .9cm of t2] (t2dots) {$\cdots$};
\node [Max,right=1.9cm of t2] (t3) {$t_i$};
\node [right=1cm of t3] (t34) {$\cdots$};

\node [Max, below=1.75cm of t0] (r0) {$r_0$};
\coordinate[below=1cm of t0] (joinhere);

\draw (s0) edge node[label] {$0$} (s1);
\draw (s1) edge node[label] {$0$} (s2);
\draw (s2) edge[dotted] node[label] {$0$} (s3);

\def\dist{-1}
\draw (s0) edge[squish] node[label] {$-2$} (t0);
\draw (s1) edge[squish] node[label] {$-4$} (t1); 
\draw (s2) edge[squish] node[label] {$-6$} (t2); 
\draw (s3) edge[squish] node[label] {$-2(i+1)$} (t3); 

\draw (t0) edge (t0dots);
\draw (t1) edge (t1dots);
\draw (t2) edge (t2dots);

\draw (t0) to node[label] {$1$} +(0,\dist) edge (r0);
\draw (t1) to node[label] {$2$} +(0,\dist) ++(0,\dist) -> (joinhere); 
 \draw (t2) to node[label] {$3$} +(0,\dist) ++(0,\dist) -> (joinhere); 
\draw[dotted] (t3) to node[label] {$i+1$} +(0,\dist) ++(0,\dist) -> (joinhere); 

\draw (r0) edge[loop right] node[label] {$0$} (r0);
\end{tikzpicture}
    }
}{\caption{The arena $\arena_4$.
        Arrows
        \inEdge{$s_i$}{$-2(i\!+\!1)$}{$t_{i}$}{squish} are shorthand for paths of length $2i+3$ with total payoff $-2(i+1)$.
        From a vertex $t_i$, \POne either exits to $r_0$ or moves to the gadget in \cref{fig:LB-3-delay}.
        \label{fig:LB-4}
        }
    }
\ffigbox{
\scalebox{0.9}{
        \tikzset{Max/.append style={minimum size=20,inner sep=0}}

\begin{tikzpicture}[node distance=1cm and 1cm]
\node [Max] (t) at (0,0) {$t_i$};
\node [Max,right=0.95cm of t] (t1) {$t_{i+1}$};
\node [Min,below=of t1] (t1') {$t_{i}^{1}$};
\node [Max,right=0.95cm of t1] (t2) {$t_{i+2}$};
\node [Min,below=of t2] (t2') {$t_{i}^{2}$};
\node [Max,right=1.25cm of t2] (tn) {$t_{i+j}$};
\node [Min,below=of tn] (tn') {$t_{i}^{j}$};
\node [right=0.35cm of tn] (t34) {$\cdots$};
\node [Max,draw=none,below=of t34] (t34') {$\cdots$};
\node [Min, below=2.5cm of t] (r0) {$r_0$};

\draw (t) edge[squish] node[label] {$i+1$} (r0);
\draw (t) edge node[label] {$1$} (t1');
\draw (t1') edge[squish] node[label,swap] {$-1$} (t1);
\draw (t2') edge[squish] node[label,swap] {$-3$} (t2);
\draw (tn') edge[squish] node[label,swap] {$-2j+1$} (tn);

\draw (t1') edge node[label] {$1$} (t2');
\draw (t2') edge[dotted] node[label] {$1$} (tn');

\end{tikzpicture}
    }    }{
 \caption{The delay gadget from vertex $t_i$ in arena $\arena_4$.
        The arrows from $t_i^j$ to $t_{i+j}$ are shorthand for paths of length
        $2j$ and payoff $-2j+1$.}
\label{fig:LB-3-delay}
}
\end{floatrow}
\end{figure}%

\begin{definition}
Let $\arena_4$ be the arena from \cref{fig:LB-4}.
It has a similar high-level structure to $\arena_3$ with different weights, and with more complex gadgets (\cref{fig:LB-3-delay}) between vertices $t_i$.
At each vertex $t_i$, \POne decides between two actions:
\begin{enumerate}
    \item to \emph{exit} to $r_0$ and gain payoff $i+1$ by doing so, or
    \item to \emph{delay} to some vertex $t_{i+j}$ where $j>0$ is chosen by \PTwo, and gain payoff $-j+1$.
\end{enumerate}
\end{definition}

Notice that, after \PTwo moved down from vertex $s_k$, \POne can (only) win by delaying at least $k + 1$ times (which we show in \cref{lem:lb-2}).
We will show that the gadgets allow \PTwo to confuse any strategy of \POne that is only based on a step counter and finite memory.
Without them, the current vertex $t_i$ together with finite extra memory would allow \POne to approximate how many delays she has chosen so far and therefore allow her to win with a finite-memory strategy.\footnote{The idea would be to partition $t_i$'s into (growing) intervals, so that each interval is picked so large that it is safe to exit from any vertex after the interval if the play entered a vertex before or at the start of that interval. A winning strategy is then to keep on delaying to $t_i$'s until vertices in three different intervals have been seen, and then exit. This requires $3$ memory states to remember the interval changes.}

A simple counting argument shows that all paths from $s_0$ to a vertex $t_k$ have the same length (proof in \cref{app:lowerBounds}).
By \cref{lem:encodingTheStep}, it implies that a step counter is useless in~$\arena_4$.%

\begin{restatable}{lemma}{sameLengthGFour} \label{lem:sameLengthG4}
For every $t_{k}$ in arena $\arena_4$,
all paths from $s_0$ to $t_k$ have the same length.
\end{restatable}

The following lemma will be used to argue that \POne wins, albeit with infinite memory.%
\begin{lemma}
\label{lem:lb-2}
    From a vertex $t_i$, if \PTwo does not stay forever in a gadget, the strategy $\strat_k$ of \POne that enters the delay gadget exactly $k\in\IN$ times achieves a total payoff of exactly $i+k+1$ in $r_0$.
\end{lemma}
\begin{proof}
    Assume that \PTwo never stays forever in a gadget (which would be winning for \POne for all quantitative objectives considered).
    The total payoff on the path from $t_i$ to the next vertex $t_{i+j}$ is $-j+1$.
    Suppose \POne delays $k$ times and let $j(1),j(2),\ldots, j(k)$
    be the lengths of the intermediate paths through gadgets, as chosen by \PTwo. That is, the play ends up in vertex $t_{i+l}$
    for $l = \sum_{c=1}^{k} j(c)$
    and has gained payoff $\sum_{c=1}^{k} \left((-j(c)+1) \right) = -l + k$.
    After $k$ delays, exiting to $r_0$ from vertex $t_{i+l}$ gives an immediate payoff of $i+l+1$.
    The total payoff from $t_i$ to $r_0$ is thus $(-l + k) + (i + l + 1) = i + k + 1$.
\end{proof}

\begin{lemma} \label{lem:insufficientTPinf}
    Consider the game played on $\arena_4$.
    Then, from vertex $s_0$,
    \begin{enumerate}
        \item \POne wins for objective $\TPinfge$; \label{enum:LB1}
        \item every step-counter and finite-memory strategy of \POne is losing for $\TPinfge$. \label{enum:LB2}
    \end{enumerate}
\end{lemma}
\begin{proof}
    For point \eqref{enum:LB1}, let $\sigma$ be the \POne strategy that, upon observing history $s_0\hops s_k\to t_k$, switches to the finite-memory strategy $\sigma_{k+1}$ from the previous lemma (delay $k+1$ times and then exit). Consider any play consistent with this strategy $\sigma$. Either \PTwo never moves to a vertex $t_k$, and then the total payoff is $0$, which is winning for \POne for $\TPinfge$. Otherwise, a vertex~$t_k$ is reached (and accordingly, the payoff until reaching it is $-2(k+1)$).
    Using $\sigma_{k+1}$, \POne guarantees a $\liminf$ total payoff of at least $0$ on any continuation: either \PTwo never leaves some gadget and the total payoff is $+\infty$, or \POne exits to $r_0$ after $k+1$ delays, which adds $k + (k + 1) + 1 = 2(k+1)$ to the total payoff by \cref{lem:lb-2}.
    In this second case, the total payoff is therefore $-2(k+1) + 2(k+1) = 0$.

    For point \eqref{enum:LB2}, by \cref{lem:sameLengthG4,lem:encodingTheStep}, it suffices to show that every finite-memory strategy of \POne is losing.
    Consider now any such strategy $\sigma_1$ of \POne with memory of size $K\in\N$ and memory update function $\memUpd$. We will show that there exists a strategy $\strat_2$ for \PTwo that is winning against $\sigma_1$.
    \PTwo's strategy is determined by 1) the initial choice of $t_j$ it visits and 2) which vertex $t_{i+j}$ to select in the gadgets (\cref{fig:LB-3-delay}) when \POne delays from vertex $t_i$. We show the existence of suitable choices by employing an argument based on the infinite Ramsey theorem, as follows.

    First, $\memUpd$ defines naturally, for any history $h\in \edges^*$,
    a function $\memUpd_h\colon \memStates \to \memStates$ that specifies how the memory is updated when observing this history (formally, $\memUpd_h(\memState) = \memUpd^*(\memState, h)$).
    Further, for every $i\ge 0$ there is a function $f_i\colon \memStates \to \{0,1\}$
    that describes for which memory states the strategy $\sigma_1$ chooses to delay or exit from $t_i$ (formally, $f_i(\memState)$ equals $1$ if $\strat_1(t_i, \memState) = (t_i, i+1, r_0)$, and $0$ otherwise).
    Since $\card{\memStates}=K\in\N$, there are only finitely many distinct such functions $f_i$ and $\memUpd_h$.
    Consider now the edge-labelled graph $G$ consisting of all vertices $t_i, i\ge0$, and where for any two $i,j\in\N$, the edge between $t_i$ and $t_{i+j}$ is labelled by the pair $(f_i,\memUpd_h)$ where $h=t_i\to t_i^1\to\cdots\to t_i^j\to t_{i+j}$ is the history through the delay gadget in $\arena_4$.

    Recall the infinite Ramsey theorem: If one labels all edges of the complete (undirected and countably infinite)
    graph with finitely many colours, then there exists an infinite monochromatic subgraph.
    Applying this to our graph $G$ yields an infinite subgraph, say with vertices $t_{\ell(i)}$ identified by $\ell\colon \N\to\N$, where all edges have the same label.
    W.l.o.g., assume that $\ell(0)\ge K$ and $\ell(i+1)>\ell(i)+1$ for all $i\ge0$.
    Based on this, the strategy $\strat_2$ of \PTwo will 1) initially move to $t_{\ell(0)}$ and 2) whenever \POne chooses to delay from $t_{\ell(i)}$ then \PTwo moves to vertex $t_{\ell(i+1)}$.
    Now consider the play $\play$ consistent with both strategies $\sigma_1$ and $\strat_2$. There are two cases.
    Either along this play \POne chooses to exit from some vertex $t_{\ell(j)}, j<K$, or not.
    If she exits too early (after delaying only $j<K$ times), then the total payoff after exiting is exactly $-2({\ell(0)+1})+ (\ell(0) + j + 1) = -\ell(0) + j - 1 $ by \cref{lem:lb-2}, which is $<0$ as $\ell(0) \ge K > j$. Hence, the play is won by \PTwo.
    Alternatively, if along the play, \POne delays at least $K$ times then, by the pigeonhole principle, there is at least one memory mode that she revisits. More precisely, the play visits vertices $t_{\ell(i)}$ and $t_{\ell(j)}$, $i<j<K$ in the same memory mode.
    Recall that the functions $f_{\ell(i)}$ are all identical for $i\ge 0$.
    It follows that the play will continue visiting vertices $t_{\ell(k)}, k\ge0$ only and never exit to $r_0$.
    Finally, observe that in any delay gadget from a vertex $t_{\ell(i)}$, the path to vertex $t_{\ell(i+1)}$
    has total payoff of $1 - (\ell(i+1) - \ell(i))$.
    Consequently, the infinite play $\play$ that visits all $t_{\ell(i)}$ will be such that $\underline{\TP}(\rho)=-\infty$
    and is losing for \POne for $\TPinfge$.
\end{proof}

\begin{theorem} \label{thm:insufficientTPinf}
    Strategies based on a step counter and finite memory are not sufficient for \POne
    in games with finitely branching arenas and objectives $\TPinfge$ or $\TPinfg$.
\end{theorem}

\begin{proof}
For $\TPinfge$ this follows directly from \cref{lem:insufficientTPinf}. For $\TPinfg$, just extend the arena by a new initial vertex $s_{-1}$ with sole outgoing edge $s_{-1}\Step{1}{}{} s_0$ to ensure that the play in which \PTwo never moves to a vertex $t_i$ is won by \POne.
\end{proof}

\section{Open objectives} \label{sec:openClosed}
The quantitative objectives defined in \cref{sec:prelim} all belong to the second or third level of the Borel hierarchy, and the strategy complexity of such objectives is not yet well understood.
However, they use as building blocks objectives from the first level of the Borel hierarchy (i.e., open and closed objectives), for which there already exist characterisations of memory requirements.
We recall some of these results for the memory structures that we study.

\subparagraph*{Step-monotonicity}
Let $\objective\subseteq\colours^\omega$ be an objective.
For two finite words $\word_1, \word_2\in\colours^*$, we write $\word_1 \prefixOrd_\objective \word_2$ if for all $\word\in\colours^\omega$, $\word_1\word\in\objective$ implies $\word_2\word\in\objective$ (meaning that the winning continuations of $\word_1$ are included in those of $\word_2$).
The relation $\prefixOrd_\objective$ is a preorder and satisfies that for $\word_1, \word_2\in\colours^*$ and $\clr\in\colours$, $\word_1\prefixOrd_\objective \word_2$ implies $\word_1\clr \prefixOrd_\objective \word_2\clr$ (i.e., it is a ``congruence'').
We write $\word_1 \strictPrefixOrd_\objective \word_2$ if $\word_1 \prefixOrd_\objective \word_2$ but $\word_2 \not\prefixOrd_\objective \word_1$.
We say that two finite words $\word_1, \word_2\in\colours^*$ are \emph{comparable for $\prefixOrd_\objective$} if $\word_1 \prefixOrd_\objective \word_2$ or $\word_2 \prefixOrd_\objective \word_1$.
We extend preorder $\prefixOrd_\objective$ to histories: we write $\history_1 \prefixOrd_\objective \history_2$ if $\col(\history_1) \prefixOrd_\objective \col(\history_2)$.

We say that an objective $\objective$ is \emph{step-monotonic} if for any two finite words $\word_1,\word_2\in\colours^*$ such that $\card{\word_1} = \card{\word_2}$, $\word_1$ and $\word_2$ are comparable for $\prefixOrd_\objective$.
In other words, for any two finite words that are read up to the same state of a step counter, one of the words must include at least the winning continuations of the other word.
This is a specialisation of the \emph{$\memory$-strong-monotony} property~\cite{BFRV23} for the step-counter memory structure $\memory = \SC$.

\begin{example} \label{ex:stepMonotonic}
    Let $\colours = \{a, b\}$.
    The open objective $\objective = aa\colours^\omega\cup bb\colours^\omega$ is \emph{not} step-monotonic, since for $\word_1 = a$ and $\word_2 = b$, we have that $\card{\word_1} = \card{\word_2}$, but $\word_1$ and $\word_2$ are not comparable for $\prefixOrd_\objective$.
    Indeed, $a^\omega$ (resp.\ $b^\omega$) is a winning continuation of $\word_1$ but not $\word_2$ (resp.\ $\word_2$ but not~$\word_1$).%

    Now, let $\colours = \IQ$ and $\step\in\IN$.
    The open objective $\objective_\step = \{\word\in\colours^\omega \mid \exists j\ge\step, \TP(\word_{\le j}) \ge 0\}$ (containing all infinite words whose total payoff goes over $0$ at some point after $\step$ steps) is step-monotonic.
    Indeed, consider two finite words $\word_1, \word_2\in\colours^*$ such that $\card{\word_1} = \card{\word_2}$.
    If $\word_2$ already satisfies $\objective_\step$ (i.e., $\word_2\colours^\omega \subseteq \objective_\step$), then necessarily, $\word_1 \prefixOrd_{\objective_\step} \word_2$.
    Similarly, if $\word_1$ already satisfies $\objective_\step$, then $\word_2 \prefixOrd_{\objective_\step} \word_1$.
    When neither $w_1$ nor $w_2$ already satisfies $\objective_\step$,  they can be compared by their current total payoff: if $\TP(\word_1) \le \TP(\word_2)$, then $\word_1 \prefixOrd_{\objective_\step} \word_2$.
    \lipicsEnd
\end{example}

\begin{remark} \label{rem:stepMonotonic}
    Variations of objective $\objective_\step$ are used as building blocks to define quantitative objectives (as can be seen in the descriptions in \cref{fig:table}), and will be considered again later.
    An important remark is that $\prefixOrd_{\objective_\step}$ is not completely determined by the current total payoff of words.
    For instance, if $\word_1 = -1, 0$ and $\word_2 = 0, -100$, we have $\word_1 \strictPrefixOrd_{\objective_1} \word_2$ even though $\TP(\word_1) > \TP(\word_2)$.
    The reason is that $\word_2$ \emph{already satisfies} $\objective_1$ after~$1$~step, and any continuation is therefore winning, despite the current total payoff being lower.
    \lipicsEnd
\end{remark}

\subparagraph*{Step-counter strategies for open objectives}
In general, the step-monotonicity property is necessary for the uniform sufficiency of step-counter strategies over finitely branching arenas (this is a specialisation of~\cite[Lemma~5.2]{BFRV23} to the step-counter memory structure $\SC$).
However, the results of~\cite{BFRV23} do not yield a characterisation for open objectives in full generality.
For the special case of the step-counter memory structure, we can actually show a converse: for open objectives, step-monotonicity implies that step-counter strategies suffice over finitely branching arenas.
This is what we show over the next three lemmas (proofs in \cref{app:openClosed}).

First, a handy result about open objectives is that in a \emph{finitely branching} arena, any winning strategy already satisfies the objective within a bounded number of steps.

\begin{restatable}{lemma}{konig} \label{lem:konig}
    Let $\objective\subseteq \colours^\omega$ be an open objective, $\arena$ be a finitely branching arena, and $\vertex_0$ be an initial vertex in $\arena$.
    If a strategy $\strat$ is winning from $\vertex_0$ for $\objective$, then there is $\step\in\IN$ such that all histories~$\history$ of length $\ge \step$ consistent with $\strat$ already satisfy $\objective$, i.e., $\col(\history)\colours^\omega \subseteq \objective$.
\end{restatable}

Second, the following lemma shows that for step-monotonic objectives, step-counter strategies can be ``locally not worse'' than arbitrary strategies.

\begin{restatable}{lemma}{locallyNotWorse} \label{lem:stepCounterCanBeNotWorse}
    Let $\objective\subseteq\colours^\omega$ be a step-monotonic objective.
    Let $\arena = \arenaFull$ be a finitely branching arena, $\vertex_0\in\vertices$ be an initial vertex, and~$\strat'$ be any strategy of \POne on $\arena$.
    There is a step-counter strategy $\strat$ such that, for every history $\history$ from $\vertex_0$ consistent with $\strat$, there is a history $\history'$ from $\vertex_0$ consistent with $\strat'$ such that $\card{\history'} = \card{\history}$, $\too(\history') = \too(\history)$, and $\history' \prefixOrd_\objective \history$.
\end{restatable}

The previous two lemmas imply that step-counter strategies suffice to win for open, step-monotonic objectives.
\begin{corollary} \label{cor:openObjectives}
    Let $\objective\subseteq\colours^\omega$ be an open, step-monotonic objective. Step-counter strategies suffice for $\objective$ over finitely branching arenas.
\end{corollary}
\begin{proof}

    Let $\arena$ be a finitely branching arena.
    Let $\vertex_0$ be a vertex from the winning region and $\strat'$ be an arbitrary winning strategy from $\vertex_0$.
    By \cref{lem:konig}, using that~$\objective$ is open and $\arena$ is finitely branching, for all histories $\history$ of length $\ge \step$ consistent with $\strat'$, we have $\col(\history)\colours^\omega \subseteq \objective$.

    As $\objective$ is step-monotonic, let $\strat$ be the step-counter strategy provided by \cref{lem:stepCounterCanBeNotWorse}.
    Every history $\history$ of length $\step$ from $\vertex_0$ consistent with $\strat$ is at least as good (for $\prefixOrd_\objective$) as a history $\history'$ of length $\step$ from $\vertex_0$ consistent with $\strat'$.
    Since $\history'$ only has winning continuations, so does $\history$.
    Therefore, strategy $\strat$ is winning from $\vertex_0$.
\end{proof}

\section{Prefix-independent \texorpdfstring{$\PiTwo$}{Pi\_2} objectives} \label{sec:PiTwo}
In this section, we show that step-counter strategies suffice for \POne for objectives $\MPsupge$ and $\TPsupinf$.
In fact, we give a sufficient condition for when step-counter strategies suffice for \POne
in finitely branching games where the objectives are prefix-independent and in $\PiTwo$.

Recall that an objective is in $\PiTwo$ if it can be written as $\bigcap_{m\in\IN} \objective_m$ for some open objectives~$\objective_m$.

\begin{theorem} \label{thm:Pi2StepCounter}
    Let $\objective = \bigcap_{m\in\IN} \objective_m \subseteq \colours^\omega$ be a prefix-independent $\PiTwo$ objective such that the objectives $\objective_m$ are open and step-monotonic.
    Then, step-counter strategies suffice uniformly for $\objective$ over finitely branching arenas.
\end{theorem}

\begin{proof}
    Let $\arena = \arenaFull$ be a finitely branching arena, and let $\vertex_0\in\vertices$ be an initial vertex.
    Let $\winning \subseteq \vertices$ be the winning region of $\arena$ for $\objective$.
    We assume that $\vertex_0$ is in the winning region $\winning$, and build a winning \emph{step-counter} strategy from~$\vertex_0$.

    We build a winning step-counter strategy $\strat\colon \vertices_1 \times \IN\to \edges$ from $\vertex_0$ by induction on parameter $m$ used in the definition of $\objective=\bigcap_{m\in\IN} \objective_m$.
    We consider the product arena $\arena \product \SC$, and fix a strategy for increasingly high step values.
    The inductive scheme is as follows: for every $m\in\IN$, we fix $\strat$ on $\vertices_1 \times \{0,\ldots, k_m - 1\}$ for some step bound $k_m\in\IN$.
    We ensure that
    \begin{itemize}
    \item along all histories from $\vertex_0$ consistent with $\strat$ of length at most $k_m$, the history does not leave $\winning$ (i.e., for all reachable $(\vertex, \step)$, we have $\vertex\in\winning$), and
    \item the open objectives $\objective_{m'}$ for $m'\le m$ are already satisfied within $k_m$ steps (i.e., any history of length $k_m$ consistent with $\strat$ only has winning continuations for $\objective_{m'}$).
    \end{itemize}
    For the base case, we may assume that we start the induction at $m = -1$ with $k_{-1} = 0$ and $\objective_{-1} = \colours^\omega$.
    We indeed have that from $(\vertex_0, 0)$, the winning region is not left within $k_{-1} = 0$ step and that the open objective $\objective_{-1}$ is already satisfied.

    Now, assume that for some $m\ge 0$, the above properties hold, so we have already fixed the moves of $\strat$ in $\arena\product\SC$ on $\vertices_1\times\{0,\ldots,k_m - 1\}$, yielding arena $(\arena \product \SC)_m$.
    We first show that in arena $(\arena \product \SC)_m$ the vertex $(\vertex_0, 0)$ still belongs to the winning region.
    We have assumed by induction that the winning region $\winning$ is not left within $k_m$ steps.
    This means that for all $(\vertex, k_m)$ reachable from $(\vertex_0, 0)$ in $(\arena \product \SC)_m$, $\vertex$ is in $\winning$.
    As $\objective$ is prefix-independent, no matter the history from $(\vertex_0, 0)$ to $(\vertex, k_m)$, there is still a winning strategy from $(\vertex, k_m)$ (recall that no choice for \POne has been fixed beyond step $k_m$).
    Hence, no matter how \PTwo plays in the first $k_m$ steps, there is still a way to win for $\objective$ from $(\vertex_0, 0)$.

    We therefore take an (arbitrary) winning strategy $\strat'_{m+1}$ of \POne from $(\vertex_0, 0)$ in $(\arena \product \SC)_m$.
    As $\strat'_{m+1}$ is winning for $\objective = \bigcap_{m\in\IN} \objective_m$, $\strat'_{m+1}$ wins in particular for the open~$\objective_{m+1}$.
    Since the arena is finitely branching and~$\objective_{m+1}$ is open, applying \cref{lem:konig}, there is $k'_{m+1}\in\IN$ such that for all histories $\history'$ of length $\ge k'_{m+1}$ consistent with $\strat_{m+1}'$, $\history'$ already satisfies $\objective_{m+1}$ (i.e., $\col(\history')\colours^\omega \subseteq \objective_{m+1}$).
    As $\objective_{m+1}$ is step-monotonic, by \cref{lem:stepCounterCanBeNotWorse}, there is a step-counter strategy $\strat_{m+1}$ such that for every history $\history$ from $(\vertex_0, 0)$ consistent with $\strat_{m+1}$, there is a history $\history'$ from~$(\vertex_0, 0)$ consistent with $\strat'_{m+1}$ such that $\card{\history'} = \card{\history}$, $\too(\history') = \too(\history)$, and $\history' \prefixOrd_{\objective_{m+1}} \history$.

    To ensure that we fix at least one extra step of the strategy in the inductive step, let $k_{m+1} = \max\{k'_{m+1}, k_m + 1\}$.
    We extend the definition of $\strat$ to play the same moves as $\strat_{m+1}$ on $\vertices_1\times\{k_m, \ldots, k_{m+1} - 1\}$, which also defines $(\arena \product \SC)_{m+1}$.
    We prove the two items of the inductive scheme.

    First, $\strat$ still does not leave $\winning$ up to step $k_{m+1}$: indeed, for every history consistent with $\strat_{m+1}$, there is a history consistent with $\strat'_{m+1}$ reaching the same vertex.
    Since $\strat'_{m+1}$ is winning and $\objective$ is prefix-independent, no such vertex can be outside of the winning region.

    Second, strategy $\strat$ then guarantees $\objective_{m+1}$ within $k_{m+1}$ steps: after $k_{m+1}$ steps, every history consistent with $\strat_{m+1}$ is at least as good for $\prefixOrd_{\objective_{m+1}}$ as a history of length $k_{m+1}$ of~$\strat'_{m+1}$.
    But every history $\history'$ of length $k_{m+1}$ consistent with $\strat'$ is such that $\col(\history')\colours^\omega \subseteq \objective_{m+1}$, and therefore has only winning continuations.

    This concludes the induction argument and shows the existence of a winning step-counter strategy from $\vertex_0$ as we iterate this process for $m\to\infty$.

    We now know that for any vertex from the winning region, there is a winning step-counter strategy.
    The existence of a \emph{uniformly} winning step-counter strategy can be shown using prefix-independence of $\objective$; this part of the proof is standard and is detailed in \cref{app:uniform}.
\end{proof}

This theorem applies to $\MPsupge$ and $\TPsupinf$ (see \cref{app:prefIndPiTwo} for a full proof).

\begin{restatable}{corollary}{MPTPsup} \label{cor:MPTPsup}
    Step-counter strategies suffice uniformly for $\MPsupge$ and $\TPsupinf$.
\end{restatable}

To illustrate \cref{thm:Pi2StepCounter} further, we apply it to a non-quantitative objective.

\begin{example}\label{example:nonQuantObj}
    Let $\colours$ be at most countable
    and $\objective\subseteq\colours^\omega$ be the objective requiring that all colours are seen infinitely often (it is an intersection of \emph{B\"uchi conditions}).
    Formally, $\objective = \bigcap_{\clr\in\colours}\bigcap_{i\ge 1}\bigcup_{j\ge i} \left\{ \word = \clr_1\clr_2\ldots \in\colours^\omega\mid \clr_j = \clr\right\}$.
    This objective is prefix-independent and in $\PiTwo$: it is the countable intersection of the open, step-monotonic objectives $\objective_{\clr, i} = \bigcup_{j\ge i} \left\{ \word = \clr_1\clr_2\ldots \in\colours^\omega\mid \clr_j = \clr\right\}$.
    By \cref{thm:Pi2StepCounter}, step-counter strategies suffice over finitely branching arenas for $\objective$.
    This result is relatively tight: finite-memory strategies do not suffice over finitely branching arenas when $\colours$ is infinite, and step-counter strategies do not suffice over infinitely branching arenas when $\card{\colours} = 2$ (see \cref{rem:example:nonQuantObj}, \cref{app:prefIndPiTwo} for details).
    \lipicsEnd
\end{example}

\section{A non-prefix-independent \texorpdfstring{$\PiTwo$}{Pi\_2} objective} \label{sec:NonPrefIndPiTwo}

In this section, we consider objective
$
    \TPsupge
    = \bigcap_{m\ge 1} \bigcap_{i \ge 1} \bigcup_{j \ge i} \left\{\word \in \colours^\omega \mid \TP(\word_{\le j}) \ge \frac{-1}{m} \right\}
$ (in $\PiTwo$).
Its definition is very close to the one of $\MPsupge$ from the previous section, but one important difference is that it is not prefix-independent (for instance, $0^\omega\in\TPsupge$, but $-1,0^\omega\notin\TPsupge$).
Hence, \cref{thm:Pi2StepCounter} does not apply.

As argued in \cref{lem:noSCforTP}, it turns out that step-counter strategies do not suffice for $\TPsupge$, even over finitely branching arenas.
We show a second example, only suited for this particular objective, illustrating more clearly the trade-off to consider to build simple winning strategies.%

\begin{figure}[t]
\centering

\adjustbox{width=\textwidth}{
\begin{tikzpicture}[node distance=1cm and 2.5cm, on grid]

\node [Min,initial,initial where=left] (v) at (0,0) {$\vertex_0$};
\node [Min,right=1.2cm of v] (v0) {$\vertex_1$};
\node [Max,right=of v0] (u0) {$u_1$};
\draw ($(v0)!0.5!(u0)-(0,.5)$) node[Min,Imp] (v00) {};
\node [Max,right=of v0] (u0) {$u_1$};
\node [Min,right=of u0] (v1) {$\vertex_2$};
\draw ($(u0)!0.5!(v1)-(0,.5)$) node[Max,Imp] (u00) {};
\node [right=1.cm of v1] (dots) {$\cdots$};
\node [Min,right=1.cm of dots] (vi) {$\vertex_i$};
\node [Max,right=of vi] (ui) {$u_i$};
\draw ($(vi)!0.5!(ui)-(0,.5)$) node[Min,Imp] (vi0) {};
\node [Min,right=of ui] (vi1) {$\vertex_{i+1}$};
\draw ($(ui)!0.5!(vi1)-(0,.5)$) node[Max,Imp] (ui0) {};

\draw (v) edge node[label,swap] {$-1$} (v0);
\draw (v0) edge[bend left=20,squish] node[label] {$0$} (u0);
\draw (v0) edge[bend right=15,squish] node[label,below left] {$1$} (v00);
\draw (v00) edge[squish,bend right=15] node[label,below] {$-2$} (u0);
\draw (u0) edge[bend left=20,squish] node[label] {$0$} (v1);
\draw (u0) edge[bend right=15,squish] node[label,below left] {$1$} (u00);
\draw (u00) edge[squish,bend right=15] node[label,below] {$-2$} (v1);
\draw (vi) edge[bend left=20,squish] node[label] {$0$} (ui);
\draw (vi) edge[squish,bend right=15] node[label,below left] {$i$} (vi0);
\draw (vi0) edge[squish,bend right=15] node[label,below] {$-i-1$} (ui);
\draw (ui) edge[bend left=20,squish] node[label] {$0$} (vi1);
\draw (ui) edge[squish,bend right=15] node[label,below left] {$i$} (ui0);
\draw (ui0) edge[squish,bend right=15] node[label,below] {$-i-1$} (vi1);
\end{tikzpicture}
}
\caption{Arena $\arena$ used in \cref{ex:bitIllustration}.
\POne has a winning $1$-bit strategy for $\TPsupge$, but no winning step-counter strategy.}
\label{fig:bitIllustration}
\end{figure}

\begin{example} \label{ex:bitIllustration}
    Consider the arena $\arena$ in \cref{fig:bitIllustration}.
    We assume that a play starts in $\vertex_0$, hence reaching sum of weights $-1$ in $\vertex_1$.
    We assume that a play is decomposed into rounds, where round $i$ corresponds to the choice of \PTwo and \POne in $\vertex_i$ and $u_i$ respectively.
    At each round $i$, \PTwo and then \POne choose either $0$, or $i$ followed by $-i-1$.
    As previously, we can assume that this arena only uses weights in $\colours = \{-1, 0, 1\}$, and that all histories from $\vertex_0$ reaching the same vertex have the same length.

    \POne has a winning strategy, consisting of playing ``the opposite'' of what \PTwo just played: if \PTwo played the sequence of $0$ (resp.\ $i, -i-1$), then \POne replies with $i, -i-1$ (resp.\ the sequence of $0$).
    This ensures that $(i)$~the current sum of weights in $\vertex_i$ is exactly $-i$ (it starts at $-1$ in $\vertex_1$ and decreases by $1$ at each round), and $(ii)$~the current sum of weights reaches exactly $0$ once during each round, after $i$ is played.
    This shows that this strategy is winning for $\TPsupge$.
    Such a strategy can be implemented with two memory states that simply remember the choice of \PTwo at each round.

    As all histories leading to vertices $u_i$ have the same length, a step-counter strategy cannot distinguish the choices of \PTwo (\cref{lem:encodingTheStep}).
    Any step-counter strategy is losing:
    \begin{itemize}
        \item either \POne only plays $0$, in which case \PTwo wins by only playing $0$, thereby ensuring that the current sum of weights is $-1$ from $\vertex_1$ onwards;
        \item or \POne plays $i, -i-1$ at some $u_i$.
        In this case, \PTwo wins by only playing $i, -i-1$.
        This means that the sum of weights decreases by at least $1$ at every round, but decreases by $2$ in round $i$.
        Hence, for $j\ge i$, the sum of weights at round $j$ is at most $-j-1$.
        Such a sum can never go above $0$ again when a player plays $j, -j-1$.
    \lipicsEnd
    \end{itemize}
\end{example}

This example shows that in general, there is a trade-off between ``obtaining a high value for a short time, to go above $0$ temporarily'' and ``playing safe in order not to decrease the value too much''.
Two memory states sufficed: if the opponent just saw a high sum of weights ($\ge 0$), then we can play it safe temporarily; if the opponent played it safe, we may need to aim for a high value, even if the overall sum decreases.
This reasoning generalises to all finitely branching arenas: in general, step-counter + $1$-bit strategies suffice for $\TPsupge$.

\begin{restatable}{theorem}{thmstepplusOneBit} \label{thm:TPsupge}
    Step-counter\:+\:$1$-bit strategies suffice for $\TPsupge$ over finitely branching arenas.
\end{restatable}

We provide a proof sketch here (full proof in \cref{app:NonPrefIndPiTwo}).
It follows the same scheme as the proof of \cref{thm:Pi2StepCounter}, where we inductively fix choices for ever longer histories.
However, we need to be more careful not to leave the winning region.
As the objective is not prefix-independent, the winning region $\winningExtended$ is described not just by a set of vertices, but by pairs of a vertex and current total payoff (i.e., the current sum of weights), i.e, $\winningExtended \subseteq \vertices \times \IQ$.

We start with a lemma about the sufficiency of memoryless strategies to \emph{stay} in this winning region.
Staying in $\winningExtended$ is necessary but not sufficient to win for $\TPsupge$.

\begin{restatable}[]{lemma}{lemsafeML} \label{lem:TPsafeML}
    Let $\arena = \arenaFull$ be a finitely branching arena.
    There exists a memoryless strategy $\strat_{\mathsf{safe}}$ of \POne in $\arena$ such that, for every $(\vertex_0, \reward)\in\winningExtended$, $\strat_{\mathsf{safe}}$ never leaves $\winningExtended$ from $\vertex_0$ with initial weight value $\reward$.
\end{restatable}

The following lemma is an analogue of \cref{lem:stepCounterCanBeNotWorse}, but ensures a stronger property with a more complex memory structure (using an extra bit).
It says that locally, with a step-counter + $1$-bit strategy, we can guarantee a high value temporarily while staying in the winning region $\winningExtended$, generalising the phenomenon of \cref{ex:bitIllustration}. The bit is used to aim for a high value (bit value $0$) or stay in the winning region (bit value $1$) by playing $\strat_{\mathsf{safe}}$ from \cref{lem:TPsafeML}.

We use a rewriting of $\TPsupge$: observe that
\begin{align}
\TPsupge = \bigcap_{m\ge 1} \bigcup_{j \ge m} \left\{\word \in \colours^\omega \mid \TP(\word_{\le j}) \ge \frac{-1}{m} \right\}, \label{eq:TPsupge}
\end{align}
where the variable $m$ is used both for the $-\frac{1}{m}$ lower bound and for the $m$ lower bound on the step count.
Indeed, this also enforces that, for arbitrarily long prefixes, the current total payoff goes above values arbitrarily close to $0$.
For $m\ge 1$, let $\objective_m = \bigcup_{j \ge m} \{\word \mid \TP(\word_{\le j}) \ge \frac{-1}{m} \}$ be the open set used in the definition of $\TPsupge$ in~\eqref{eq:TPsupge}.

\begin{restatable}[]{lemma}{lemTPsupgeHard}
 \label{lem:TPsupgeHard}
    Let $\arena = \arenaFull$ be an arena and $\vertex_0\in\vertices$ be an initial vertex in the winning region of \POne for $\TPsupge$.
    For all $m\ge 1$, there exists a step-counter + $1$-bit strategy~$\strat_m$ such that $\strat_m$ is winning for $\objective_m$ from $\vertex_0$ and never leaves $\winningExtended$ (i.e., for all histories $\history$ from $\vertex_0$ consistent with $\strat_m$, $(\too(\history), \TP(\history))\in\winningExtended$).
\end{restatable}

The inductive scheme used in the proof of \cref{thm:TPsupge} is similar to that of \cref{thm:Pi2StepCounter}, building a step-counter + 1-bit strategy $\strat\colon \vertices_1 \times \IN \times \{0, 1\} \to \edges$.

For $\memory$ a step-counter and $1$-bit memory structure, consider the product arena $\arena' = \arena\product \memory$ (in which the bit updates are not fixed yet, and will be fixed inductively).
    We have that $(\vertex_0, (0, 0))$ is in the winning region of $\arena'$.
    The inductive scheme is as follows: for infinitely many $m\in\IN$, for some step bound $k_m\in\IN$, we fix $\strat$ on $\vertices_1 \times \{0,\ldots, k_m - 1\} \times \{0, 1\}$, yielding arena $\arena'_m$.
    Using \cref{lem:TPsupgeHard}, we ensure that
    \begin{itemize}
    \item along all histories $\history$ from $\vertex_0$ consistent with $\strat$ of length at most $k_m$, $\winningExtended$ is not left, and
    \item the open objective $\objective_{m}$ is already satisfied within $k_m$ steps (i.e., any history of length $k_m$ consistent with $\strat$ only has winning continuations for $\objective_{m}$).
    \end{itemize}
Iterating this procedure defines a step-counter + $1$-bit strategy $\strat$ that satisfies $\objective_m$ for infinitely many $m\ge 1$.
Since the sequence $(\objective_m)_{m\ge 1}$ is decreasing ($\objective_1 \supseteq \objective_2 \supseteq \ldots$), we have that $\strat$ is winning for $\objective_m$ for all $m\ge 1$.
Hence, $\strat$ is winning for $\TPsupge$.

\begin{remark}
    Unlike for \cref{thm:Pi2StepCounter}, the upper bound in this section does not apply \emph{uniformly} in general (an arena illustrating this is in \cref{app:NonPrefIndPiTwo}, \cref{fig:nonUniformPiTwo}).
    \lipicsEnd
\end{remark}

\begin{remark}
Over integer weights ($\colours\subseteq\IZ$), $\TPsupg = \overline{\TP}_{\ge 1} \in \PiTwo$.
    As $\overline{\TP}_{\ge 1}$ behaves like $\TPsupge$ (\cref{rem:threshold0}), the results from this section apply to $\TPsupg$ over integer weights.
    Up to some scaling factor, this also applies to rational weights with bounded denominators.
    However, for general rational weights, $\TPsupg$ can only be shown to be in $\SigmaThree$, so the above does not apply.
    \lipicsEnd
\end{remark}

\section{Conclusion}
We established whether step-counter strategies (possibly with finite memory) suffice for the objectives $\MPsupge$, $\TPinfg$, $\TPinfge$, $\TPsupinf$, and $\TPsupge$.
We used the structure of these objectives as sets in the Borel hierarchy, and pinpointed the strategy complexity for all classical quantitative objectives on the second level of Borel hierarchy.
This leaves open the cases of $\MPsupg$, $\MPinfge$, $\TPsupg$ (over $\IQ$), and $\TPinfinf$, all on the third level.
The sufficiency of other less common infinite memory structures, such as \emph{reward counters}~\cite{MM23}, could also be investigated.

\bibliography{stepCounters.bib}
\newpage

\appendix

\section{Missing proofs for Section~\ref{sec:lowerBounds}} \label{app:lowerBounds}

We give formal proofs of \cref{lem:infBranchingLB,lem:sameLengthG4}.

\infBranchingLB*
\begin{proof}
Let $\arena_1$, depicted in \cref{fig:LB-0a}, be the arena with three vertices, $s$ controlled by \PTwo and $t,q$ controlled by \POne, with weights $\colours=\IZ$ and edges $E = (\{s\}\x(-\N)\x\{t\}) \cup (\{t\}\x\N\x\{q\})\cup \{(q,0,q)\}$. In this arena, \POne has a strategy to win for any total-payoff objective with threshold $0$: given the choice of $-i$ by \PTwo, \POne can respond with $i+1$, thus winning for $\TPinfg$, $\TPinfge$, $\TPsupg$, and $\TPsupge$.
Notice that, at $t$, the step-counter value is always $1$, providing no useful information.
Consider any finite-memory strategy of \POne. There must be a maximum number $k$ chosen by this strategy at $t$, against any possible choices of \PTwo. Against such a strategy, \PTwo wins by playing $-k-1$.

We now show the claim for $\MPsupg$, $\MPsupge$, and $\TPsupinf$. Consider $\arena_1'$ (see \cref{fig:LB-0b}), which repeats the process of $\arena_1$. Vertices $t$ and $s$ are controlled by \POne and \PTwo respectively, with weights $\colours=\IZ$ and edges $E = (\{s\}\x(-\N)\x\{t\}) \cup (\{t\}\x\N\x\{s\})$.
That is, the structure enforces strict alternation and each step has an arbitrary finite weight.

Notice that regardless of the players' choices, the number of steps up to round~$i$ is always~$2i$.
Every \POne strategy based on a step counter and finite memory thus defines a function $f\colon\N\to\N$ where $f(i)$ denotes the maximum number chosen in round~$i$. \PTwo can counter such a strategy by always picking $-f(i) - 1$ in round $i$ and thereby ensure a mean payoff $\le -\frac{1}{2}$ and a total payoff of~$-\infty$.
\end{proof}

\sameLengthGFour*
\begin{proof}
    By induction on $k$.
    We show that the length of any path from $s_0$ to $t_k$ is $3(k+1)$.
    For $k=0$, there is only one path of length $3=3(k+1)$.
    For $k+1$, the ``direct'' path $s_0\hops s_{k+1}\to t_{k+1}$ evidently has the stated length: $k+1$ steps to reach $s_{k+1}$ then $2(k+1)+3$ steps down to $t_{k+1}$.
    Consider any other path $s_0\hops t_{j}\hops t_{k+1}$
    with $j<k+1$ maximal, i.e., the suffix $t_j\hops t_{k+1}$
    went through the delay gadget from $t_j$. By induction hypothesis, the prefix up to vertex $t_j$ has length $l_1 = 3(j+1)$. The suffix from $t_j\to t_j^1\hops t_j^{k+1-j}\to t_{k+1}$ has length
    $l_2=(k+1-j) + 2(k+1-j)$.
    The total length of the path is thus $l_1+l_2 = 3(j+1) + 3(k+1-j) = 3(k+2)$ as required.
\end{proof}

\section{Missing proofs for Section~\ref{sec:openClosed}}
\label{app:openClosed}

In this section, we prove \cref{lem:stepCounterCanBeNotWorse,lem:konig}.

\konig*
\begin{proof}
    Let $\strat$ be a strategy winning from $\vertex_0$ for $\objective$.
    Let $T_\strat$ be the set of all histories $\history$ from~$\vertex_0$ consistent with $\strat$ such that $\col(\history)\colours^\omega \not\subseteq \objective$.
    We have that $T_\strat$ is a tree, since if $\col(\history)\colours^\omega \subseteq \objective$, then for $\history'$ a longer history with $\history$ as a prefix, we also have $\col(\history')\colours^\omega \subseteq \objective$.

    Assume that, for all $\step\in\IN$, there is~$\history$ of length $\ge \step$ such that $\col(\history)\colours^\omega \not\subseteq \objective$.
    Then, $T_\strat$ is an infinite tree.
    By K\H{o}nig's lemma, this tree has an infinite branch.
    Hence, there is an infinite play $\play$ consistent with $\strat$ such that all finite prefixes $\history$ of $\play$ are such that $\col(\history)\colours^\omega \not\subseteq \objective$.
    As~$\objective$ is open, this means that $\col(\play) \not\in\objective$, so $\strat$ is not winning from $\vertex_0$.
\end{proof}

\locallyNotWorse*
\begin{proof}
    We define $\strat\colon \vertices_1 \times \IN \to \edges$ and prove the required property on $\strat$ by induction on $\IN$.
    The property is trivially true for histories of length $0$ (i.e., just fixing an initial vertex).

    We assume that $\strat$ has been defined on $\vertices_1 \times \{0, \ldots, \step - 1\}$ for some $\step \ge 0$, and that the property holds for histories up to length $\step$.
    We define $\strat$ on histories of length $\step$ and prove the properties on histories of length~$\step + 1$ at the same time.

    Let $\history$ be any history from $\vertex_0$ consistent with $\strat$ of length $\step$.
    Let $\vertex = \too(\history)$.
    By induction hypothesis, there is a history $\history'$ from $\vertex_0$ consistent with $\strat'$ such that $\card{\history'} = \step$, $\too(\history') = \vertex$, and~$\history' \prefixOrd_\objective \history$.

    If $\vertex\in\vertices_2$, we consider any outgoing edge $\edge$ of $\vertex$ (which could be a move played by \PTwo after $\history$).
    The history $\history\edge$ has length $\step+1$ and is consistent with $\strat$.
    Since $\history'$ also ends in $\vertex\in\vertices_2$, the history $\history'\edge$ is also consistent with $\strat'$.
    As we had $\history' \prefixOrd_\objective \history$, we also have $\history'\edge \prefixOrd_\objective \history\edge$ (using that $\prefixOrd_\objective$ is a congruence).
    Hence, history $\history'\edge$ satisfies all the required properties: it is consistent with $\strat'$, it is of length $\step+1$, it ends in $\too(\history\edge)$, and it is such that $\history'\edge \prefixOrd_\objective \history\edge$.

    If $\vertex\in\vertices_1$, we need to define $\strat(\vertex, \step)$ to see how history $\history$ is extended.
    For $\vertex_1\in\vertices$, we define a history $\history_{\vertex, \step}'$ as one of the elements of
    \[
        \min_{\prefixOrd_\objective}
        \{\history''\in\histories(\arena) \mid \text{$\history''$ is consistent with $\strat'$, $\from(\history'') = \vertex_0$, $\card{\history''} = \step$, and $\too(\history'') = \vertex$}\}.
    \]
    The set being minimised over is non-empty, as $\history'$ is in it.
    It is also finite, as $\arena$ is finitely branching and there are therefore only finitely many histories of fixed length from $\vertex_0$.
    Moreover, a minimum exists as all histories of the same length are comparable for~$\prefixOrd_\objective$ due to step-monotonicity.
    Since $\prefixOrd_\objective$ is a \emph{pre}order, there may be multiple histories that are minimal but equivalent with respect to~$\prefixOrd_\objective$, in which case we can pick any of them.

    We define $\strat(\vertex, \step) = \strat'(\history_{\vertex, \step}')$.
    Let $\history\edge$ be the one-move continuation of $\history$, where $\edge = \strat(\vertex, \step)$.
    We have that $\history_{\vertex, \step}'\edge$ is consistent with $\strat'$, and satisfies $\card{\history_{\vertex, \step}'\edge} = \card{\history\edge}$ and $\too(\history_{\vertex, \step}'\edge) = \too(\history\edge)$.
    Moreover, we had $\history'\prefixOrd_\objective \history$ (by the induction hypothesis) and we have $\history_{\vertex, \step}' \prefixOrd_\objective \history'$ by the choice of minimum.
    Therefore, $\history_{\vertex, \step}' \prefixOrd_\objective\history$, so $\history_{\vertex, \step}'\edge \prefixOrd_\objective \history\edge$, which ends the proof.
\end{proof}

\section{Uniformly winning strategies} \label{app:uniform}

In this section, we prove that winning step-counter strategies can be ``uniformised'' for prefix-independent objectives.
This result easily follows from the following known result~\cite[Lemma~5]{CN06} on the uniformisation of \emph{memoryless} strategies for prefix-independent objectives.
We rephrase this lemma with our notations.
This result was originally stated for both players at the same time, but its proof applies to one player at a time.

\begin{lemma}[{\cite[Lemma~5]{CN06}}] \label{lem:uniformML}
    Let $\objective$ be a prefix-independent objective, and $\arena = \arenaFull$ be an arena.
    If, for all $\vertex\in\winning$ in the winning region of \POne, there is a winning memoryless strategy from $\vertex$, then \POne has a uniformly winning memoryless strategy in~$\arena$.
\end{lemma}

We obtain a similar result for step-counter strategies by reducing to the memoryless case through the construction of the product arena $\arena \product \SC$.

\begin{lemma} \label{lem:uniform}
    Let $\objective$ be a prefix-independent objective and $\arena = \arenaFull$ be an arena in which, from every $\vertex\in\winning$ of the winning region of \POne, \POne has a winning step-counter strategy.
    Then \POne has a uniformly winning step-counter strategy.
\end{lemma}
\begin{proof}
    Consider the product arena $\arena \product \SC$.
    The fact that, from every vertex $\vertex\in\winning$, \POne has a winning step-counter strategy is equivalent to the fact that, from every $(\vertex, 0)$ in $\arena \product \SC$ for $\vertex\in\winning$, \POne has a winning memoryless strategy.
    The proof of this equivalence is standard; full details can be found in~\cite[Lemma~2.4]{BLORV22}.

    As $\objective$ is prefix-independent, we can apply \cref{lem:uniformML} to arena $\arena\product\SC$ and find a memoryless strategy $\strat$ of \POne that wins from all vertices in the winning region of $\arena\product\SC$.
    In particular, it wins from all vertices $(\vertex, 0)$ with $\vertex\in\winning$.
    Going back to $\arena$, we find that there is a step-counter strategy that wins from all $\vertex\in\winning$.
\end{proof}

\section{Missing details for Section~\ref{sec:PiTwo}} \label{app:prefIndPiTwo}

We prove \cref{cor:MPTPsup}.

\MPTPsup*
\begin{proof}
    Let $\colours = \IQ$.
    Recall that both $\MPsupge$ and $\TPsupinf$ are prefix-independent.

    We first focus on $\MPsupge$.
    A natural way to write $\MPsupge$ as a $\PiTwo$ objective, starting from its definition, is as the set of infinite words that have infinitely many finite prefixes with mean payoff above~$\frac{-1}{m}$ for all $m\ge 1$. Formally,
    \[
        \MPsupge = \bigcap_{m\ge 1} \bigcap_{i \ge 1} \bigcup_{j \ge i} \left\{\word \in\colours^\omega\mid \MP(\word_{\le j}) \ge \frac{-1}{m}\right\}.
    \]
    For $m, i\ge 1$, let $\objective_{m, i} = \bigcup_{j \ge i} \{\word \mid \MP(\word_{\le j}) \ge \frac{-1}{m}\}$.
    Such sets are open, as whether a word belongs to it is witnessed by a finite prefix.
    Observe that the objectives $\objective_{m, i}$ are also step-monotonic. Indeed, if we consider two finite words of the same length, either one has already satisfied $\objective_{m, i}$, or the one with the greatest current mean payoff has more winning continuations than the other (a similar reasoning was used in \cref{ex:stepMonotonic}).

    Therefore, $\MPsupge$ is the countable intersection of open, step-monotonic objectives.
    By \cref{thm:Pi2StepCounter}, we conclude that a step counter suffices for $\MPsupge$ in finitely branching arenas.

    The claim for $\TPsupinf$ can be shown analogously: observe that
    \[
        \TPsupinf
        = \bigcap_{m\ge 1} \bigcap_{i \ge 1} \bigcup_{j \ge i} \left\{\word\in\colours^\omega \mid \TP(\word_{\le j}) \ge m \right\}
    \]
    is in $\PiTwo$ and fix $\objective_{m, i} = \bigcup_{j \ge i} \left\{\word\in\colours^\omega \mid \TP(\word_{\le j}) \ge m \right\}$.
\end{proof}

\begin{remark}[Further details on \cref{example:nonQuantObj}]\label{rem:example:nonQuantObj}
    The $\objective_{\clr, i}$ are open because whether an infinite word is in such an objective is witnessed by a finite prefix of the word.
    Also, the objectives $\objective_{\clr, i}$ are step-monotonic: if we take two finite words of the same length, either one already satisfies the objective (in which case, this word is larger for $\prefixOrd_{\objective_{\clr, i}}$), or none of the words has already satisfied the objective, in which case both words have exactly the same winning continuations.

    We show that finite-memory strategies do not suffice over finitely branching arenas when~$\colours$ is infinite: consider the arena in \cref{fig:tightInfiniteBuchi-a}.
    In this arena, \POne can see all colours infinitely often, but any finite-memory strategy either only stays on the top horizontal line, or never goes beyond some $\clr_i$.

    We show that step-counter strategies do not suffice over infinitely branching arenas when $\card{\colours} = 2$: consider the arena in \cref{fig:tightInfiniteBuchi-b}.
    Clearly, \POne wins by playing $\clr_2$ followed by $\clr_1$ every time the play reaches $v$.
    However, using only a step-counter strategy $\strat$, there needs to be infinitely many $\step$ such that $\strat(v, \step) = (v, \clr_2, v)$ and infinitely many $\step$ such that $\strat(v, \step) = (v, \clr_1, u)$.
    For any given step-counter strategy, \PTwo can therefore have a winning counterstrategy by ensuring that the play reaches $v$ only in those steps where \POne plays $(v, \clr_1, u)$ immediately.
    \lipicsEnd
\end{remark}

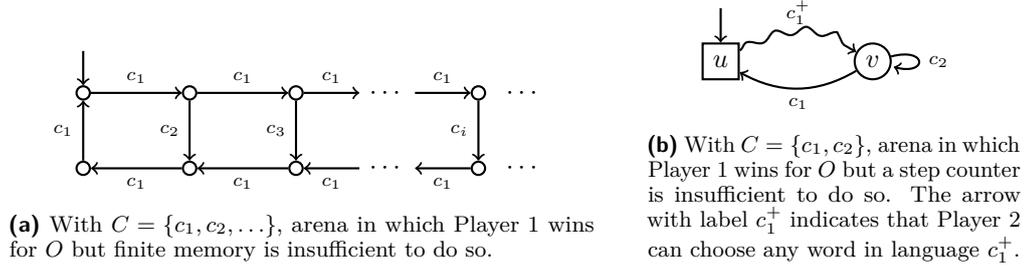
\begin{figure}
    \centering
    \begin{subfigure}[b]{0.55\textwidth}
    \centering
    \begin{tikzpicture}[node distance=1cm and 2cm, on grid]

\node [Max,Imp,initial,initial where=above] (s0) at (0,0) {};
\node [Max,Imp,right=1.4cm of s0] (s1) {};
\node [Max,Imp,right=1.4cm of s1] (s2) {};
\node [right=1.2cm of s2] (s23) {$\cdots$};
\node [Max,Imp,right=1.2cm of s23] (s3) {};
\node [right=.6cm of s3] (s34) {$\cdots$};

\node [Max, Imp, below=of s0] (s0') {};
\node [Max, Imp, below=of s1] (s1') {};
\node [Max, Imp, below=of s2] (s2') {};
\node [below=of s23] (s23') {$\cdots$};
\node [Max, Imp, below=of s3] (s3') {};
\node [below=of s34] (s34') {$\cdots$};

\draw (s0) edge node[label] {$\clr_1$} (s1);
\draw (s1) edge node[label] {$\clr_1$} (s2);
\draw (s2) edge node[label] {$\clr_1$} (s23);
\draw (s23) edge node[label] {$\clr_1$} (s3);

\draw (s0') edge node[label,auto] {$\clr_1$} (s0);
\draw (s1) edge node[label,auto,swap] {$\clr_2$} (s1');
\draw (s2) edge node[label,auto,swap] {$\clr_3$} (s2');
\draw (s3) edge node[label,auto,swap] {$\clr_i$} (s3');
\draw (s3') edge node[label,auto] {$\clr_1$} (s23');
\draw (s23') edge node[label,auto] {$\clr_1$} (s2');
\draw (s2') edge node[label,auto] {$\clr_1$} (s1');
\draw (s1') edge node[label,auto] {$\clr_1$} (s0');

\end{tikzpicture}
    \caption{With $\colours = \{\clr_1, \clr_2, \ldots\}$, arena in which \POne wins for $\objective$ but finite memory is insufficient to do so.}
    \label{fig:tightInfiniteBuchi-a}
    \end{subfigure}%
    \hspace{0.05\textwidth}%
    \begin{subfigure}[b]{0.35\textwidth}
    \centering
    \begin{tikzpicture}[node distance=1cm and 2cm, on grid]

\node [Min,initial,initial where=above] (s0) at (0,0) {$u$};
\node [Max,right=2cm of s0] (s1) {$v$};

\draw (s0) edge[bend left,squish] node[label] {$\clr_1^+$} (s1);
\draw (s1) edge[bend left] node[label] {$\clr_1$} (s0);
\draw (s1) edge[loop right] node[label] {$\clr_2$} (s1);

\end{tikzpicture}
    \caption{With $\colours = \{\clr_1, \clr_2\}$, arena in which \POne wins for $\objective$ but a step counter is insufficient to do so.
    The arrow with label $\clr_1^+$ indicates that \PTwo can choose any word in language $\clr_1^+$.}%
    \label{fig:tightInfiniteBuchi-b}
    \end{subfigure}
    \caption{Arenas illustrating that memory is insufficient for objective $\objective$ in \cref{rem:example:nonQuantObj}.}
\end{figure}

\section{Missing proofs for Section~\ref{sec:NonPrefIndPiTwo}} \label{app:NonPrefIndPiTwo}
This section is dedicated to the proof of \cref{thm:TPsupge}.

\thmstepplusOneBit*

For an arena $\arena = \arenaFull$, let $\winningExtended \subseteq \vertices \times \IQ$ be the set of \emph{pairs} $(\vertex, \reward)$ such that there is a winning strategy from $\vertex$ for $\TPsupge$, assuming the current sum of weights is $\reward$.
Formally, $(\vertex, \reward)\in\winning'$ if there is a strategy $\strat$ such that for all plays $\play$ consistent with $\strat$, $\reward + \limsup_j(\TP(\col(\play)_{\le j})) \ge 0$.
For example, in \cref{ex:bitIllustration} (\cref{fig:bitIllustration}), we have $(\vertex_i, -i)\in\winningExtended$, but $(\vertex_i, -i-1)\notin\winningExtended$.
We say that a strategy \emph{never leaves $\winningExtended$ from $\vertex_0$ with initial weight value $r$} if for all histories $\history$ from $\vertex_0$ consistent with the strategy, $(\too(\history), r + \TP(\history))\in\winningExtended$.

We start by proving \cref{lem:TPsafeML} about the sufficiency of memoryless strategies to \emph{stay} in this winning region.
Note that staying in $\winningExtended$ is necessary but not sufficient for a strategy to be winning for $\TPsupge$.

\lemsafeML*
\begin{proof}
    We define a memoryless strategy $\strat_{\mathsf{safe}}\colon \vertices_1 \to \edges$ with the required properties.
    Let $\vertex\in\vertices_1$.
    Assume that the set $R_\vertex = \{\reward\in\IQ \mid (\vertex, \reward)\in\winningExtended\}$ is not empty (if it is, we define $\strat(\vertex)$ arbitrarily).
    Notice that $R_\vertex$ is upwards closed (i.e., if $\reward\in R_\vertex$ and $\reward' \ge \reward$, then $\reward'\in R_\vertex$).
    Let $r_v = \inf R_v$.
    If $r_v\in R_v$, there exists a strategy $\sigma^v$ that wins for $\TPsupge$ from $\vertex$ with current weight value $r_v$; we fix $\strat_{\mathsf{safe}}(\vertex) = \strat^v(\emptyHistory_\vertex)$ (we recall that $\emptyHistory_\vertex$ is the empty history from $\vertex$).
    If $r_v\notin R_v$, for every $n \ge 1$, as $(\vertex, r_v + \frac{1}{n})\in\winningExtended$, there exists a strategy $\sigma_n^v$ that wins for $\TPsupge$ from $\vertex$ with current weight value $r_v + \frac{1}{n}$.
    Consider the infinite sequence of edges $\strat_1^v(\emptyHistory_\vertex), \strat_2^v(\emptyHistory_\vertex), \ldots$; as $\arena$ is finitely branching, one of the edges outgoing from $\vertex$ appears infinitely often along this sequence.
    We define $\strat_{\mathsf{safe}}(\vertex)$ to be such an edge.

    We now show that $\strat_{\mathsf{safe}}$ satisfies the property from the statement.
    Let $(\vertex_0, r)\in\winningExtended$.
    Let $\history$ be a history from $\vertex_0$ consistent with $\strat_{\mathsf{safe}}$.
    We show by induction on the length of $\history$ that $\reward + \TP(\history) \in R_{\too(\history)}$.
    For the base case, we have that $\reward \in R_{\vertex_0}$.

    Assume now that $\history\edge$ is a history from $\vertex_0$ consistent with $\strat_{\mathsf{safe}}$ of length $\card{\history} + 1$ and that $\reward + \TP(\history) \in R_{\too(\history)}$.
    We have that $\inf R_{\too(\history)} \le \reward + \TP(\history)$.
    If $\inf R_{\too(\history)} \in R_{\too(\history)}$, then $\inf R_{\too(\history)} + \col(\edge) \in R_{\too(\edge)}$ by definition of $\strat_{\mathsf{safe}}(\too(\history))$, so $\reward + \TP(\history\edge) \in R_{\too(\edge)}$ (as $R_{\too(\edge)}$ is upwards closed).
    If $\inf R_{\too(\history)} \notin R_{\too(\history)}$, then $\inf R_{\too(\history)} < \reward + \TP(\history)$.
    So we can find $n\in\IN$ such that $\inf R_{\too(\history)} + \frac{1}{n} \le \reward + \TP(\history)$ and such that $\strat_{\mathsf{safe}}(\too(\history)) = \strat^{\too(\history)}_n(\history_{\too(\history)})$.
    With a similar observation as the previous case, we obtain that $\reward + \TP(\history\edge) \in R_{\too(\edge)}$.
\end{proof}

\cref{lem:TPsupgeHard} is an analogue of \cref{lem:stepCounterCanBeNotWorse}, but ensures a stronger property with a more complex memory structure (using an extra bit).
It says that locally, with a step-counter + $1$-bit strategy, we can guarantee a high value temporarily while staying in the winning region $\winningExtended$.
It generalises the phenomenon of \cref{ex:bitIllustration}, in which we observed that $1$ bit is exactly what we need to either stay in the winning region or aim for a high value.
We will later use it in the induction step in the proof of \cref{thm:TPsupge}.

As explained in \cref{sec:NonPrefIndPiTwo}, we recall that
\begin{align*}
\TPsupge = \bigcap_{m\ge 1} \bigcup_{j \ge m} \left\{\word \in \colours^\omega \mid \TP(\word_{\le j}) \ge \frac{-1}{m} \right\},
\end{align*}
and that for $m\ge 1$, we let $\objective_m = \bigcup_{j \ge m} \{\word \mid \TP(\word_{\le j}) \ge \frac{-1}{m} \}$.
To better understand $\prefixOrd_{\objective_m}$, recall that as in \cref{rem:stepMonotonic}, for any two words $\word_1, \word_2\in\colours^*$ with $\card{\word_1} = \card{\word_2}$, we have $\word_1\prefixOrd_{\objective_m}\word_2$ if and only if
\begin{align} \label{eq:prefixOrdOm}
    \word_2\colours^\omega \subseteq \objective_m\quad \text{or}\quad (\word_1\colours^\omega \not\subseteq \objective_m\ \text{and}\ \TP(\word_1) \le \TP(\word_2)).
\end{align}
The first term of the disjunction corresponds to the case where $\word_2$ has a prefix that already satisfies $\objective_m$, and hence every continuation is winning for $\objective_m$.
The second term then covers the case where none of the two words has satisfied $\objective_m$ so far and $\word_2$ has a current sum of weights at least as high as $\word_1$.

\lemTPsupgeHard*
\begin{proof}
Let $\strat'$ be an arbitrary winning strategy from $\vertex_0$ (which exists as $\vertex_0$ is in the winning region, i.e., $(\vertex_0, 0) \in \winningExtended$).
Let $\strat_{\mathsf{safe}}$ be a memoryless strategy on $\arena$ staying in $\winningExtended$ given by \cref{lem:TPsafeML}.
We consider the ``step-counter + $1$-bit'' memory structure $\memory = (\IN \times \{0, 1\}, (0, 0), \memUpd)$ (we have not yet defined how and when $\memUpd$ updates the bit, which will come later in the proof).

We define the strategy $\strat_m$ based on $\memory$ inductively from $\vertex_0$, on the set $\vertices_1 \times \IN \times \{0, 1\}$.
As defined above, the initial memory bit is $0$.
The goal is to guarantee the following two properties: for all histories $\history$ consistent with $\strat_m$ from $\vertex_0$,
{
\begin{enumerate}
    \item we have $(\too(\history), \TP(\history)) \in \winningExtended$, that is, the strategy does not leave the winning region for $\TPsupge$; \label{enum:1}
    \item if the memory bit after $\history$ is $0$, there exists a history $\history'$ consistent with $\strat'$ from $\vertex_0$ such that $\card{\history'} = \card{\history}$, $\too(\history') = \too(\history)$, and $\history' \prefixOrd_{\objective_m} \history$. \label{enum:2}
\end{enumerate}}
The properties clearly hold on the empty history $\emptyHistory_{\vertex_0}$ from $\vertex_0$ (for which the bit value is still~$0$).
We will prove the two properties inductively after defining the strategy $\strat_m$.%
\footnote{As a side note, we briefly comment on how these two properties can be interpreted on the game of \cref{ex:bitIllustration} (\cref{fig:bitIllustration}).
Observe that following the $0$ edges guarantees~\ref{enum:1}.\ for \POne (it ensures a higher sum of weights), but does not guarantee~\ref{enum:2}.\ (as $\history$ may not have already satisfied one of the $\objective_m$, unlike all histories $\history'$).
On the other hand, the first time \POne plays $i, -i-1$, it ensures~\ref{enum:2}.\ (some $\objective_m$ is already satisfied during this round), but not~\ref{enum:1}.\ (it may leave the winning region $\winningExtended$ if \PTwo also played $i, -i-1$).
This is why the bit is necessary to guarantee both~\ref{enum:1}.\ and~\ref{enum:2}.}

Assume that $\strat_m$ has already been defined on $\vertices_1 \times \{0,\ldots,\step-1\} \times \{0,1\}$.
Let $\history$ be a history from $\vertex_0$ consistent with $\strat_m$ of length $\step$.
We write $\vertex$ for $\too(\history)$ for brevity.
To extend~$\strat_m$, we need to define $\strat_m(\vertex, \step, 0)$ and $\strat_m(\vertex, \step, 1)$ for $\vertex\in\vertices_1$, and we need to define when the memory bit is updated.
If the memory bit after $\history$ is $0$, we define a history
\[
    \history'_{\vertex, \step} \in \min_{\prefixOrd_{\objective_m}}
    \{\history' \mid \text{$\history'$ is consistent with $\strat'$, $\from(\history') = \vertex_0$, $\card{\history'} = \step$, and $\too(\history') = \vertex$}\}.
\]
This set is non-empty due to the induction hypothesis~\ref{enum:2}., and well-ordered due to step-monotonicity of $\objective_m$, so $\history'_{\vertex, \step}$ is well-defined.
If $\vertex\in\vertices_1$, we define $\strat_m(\vertex, \step, 0) = \strat'(\history'_{\vertex, \step})$ and $\strat_m(\vertex, \step, 1) = \strat_\mathsf{safe}(\vertex)$.

We now define the bit update.
Let $\edge$ be a possible edge taken from $\vertex$ after history $\history$ (either $\edge$ is the edge taken by $\strat_m$ after $\history$ if $\vertex\in\vertices_1$, or $\edge$ is any possible outgoing edge of $\vertex$ if $\vertex\in\vertices_2$).
When the bit is $1$, we never change the bit (i.e., $\memUpd((\step, 1), \edge) = (\step+1, 1)$).
Assume now that the bit of $\strat_m$ is $0$ after $\history$.
We update the bit from $0$ to $1$ (i.e., define $\memUpd((\step, 0), \edge) = (\step+1, 1)$) if and only if $\col(\history'_{\vertex, \step}\edge)\colours^\omega\subseteq\objective_m$.
Intuitively, this condition means that $\strat'$ has already satisfied $\objective_m$, which will be used to show that so does $\strat_m$.

Now that we have fully defined the next step of $\strat_m$ after $\history$, it is left to prove the two properties.
To do so, we first prove the following property:
\begin{align} \label{eq:unusualPreorder}
\text{if the memory bit after $\history$ is $0$, then $\TP(\history_{\vertex, \step}') \le \TP(\history)$.}
\end{align}
We prove the contrapositive.
Assume that $\TP(\history) < \TP(\history_{\vertex, \step}')$.
Observe that this cannot happen if $\card{\history} = 0$.
Let $\widetilde{\history}\edge$ be the shortest prefix of $\history$ such that $\TP(\history_{\too(\widetilde{\history}), \card{\widetilde{\history}}}') \le \TP(\widetilde{\history})$ but $\TP(\widetilde{\history}\edge) < \TP(\history_{\too(\widetilde{\history}\edge), \card{\widetilde{\history}\edge}}')$.
We have that $\history_{\too(\widetilde{\history}), \card{\widetilde{\history}}}'\edge$ is consistent with $\strat'$ and $\TP(\history_{\too(\widetilde{\history}), \card{\widetilde{\history}}}'\edge) \le \TP(\widetilde{\history}\edge)$.
Combining the last two inequalities, we obtain that $\TP(\history_{\too(\widetilde{\history}), \card{\widetilde{\history}}}'\edge) < \TP(\history_{\too(\widetilde{\history}\edge), \card{\widetilde{\history}\edge}}')$.
However, by definition of $\history_{\too(\widetilde{\history}\edge), \card{\widetilde{\history}\edge}}'$, we have that $\history_{\too(\widetilde{\history}\edge), \card{\widetilde{\history}\edge}}' \prefixOrd_{\objective_m} \history_{\too(\widetilde{\history}), \card{\widetilde{\history}}}'\edge$.

Using the discussion about $\prefixOrd_{\objective_m}$ from~\eqref{eq:prefixOrdOm}, this necessarily implies that $\col(\history_{\too(\widetilde{\history}), \card{\widetilde{\history}}}'\edge)\colours^\omega \subseteq \objective_m$.
So the bit was set to $1$ after $\widetilde{\history}\edge$, which means that the bit is still $1$ after $\history$.

We now prove the two inductive properties.
{
\begin{enumerate}
    \item By induction hypothesis, we have $(\too(\history), \TP(\history)) \in \winningExtended$.

    If $\vertex\in\vertices_2$, then for any possible edge $\edge$ from $\vertex$, we still have $(\too(\history\edge), \TP(\history\edge)) \in \winningExtended$ by definition of $\winningExtended$.
    We now assume $\vertex\in\vertices_1$ and distinguish whether the bit value after $\history$ is~$0$ or $1$.

    If the bit is $0$ after $\history$, then by~\eqref{eq:unusualPreorder}, $\TP(\history_{\vertex, \step}') \le \TP(\history)$.
    Let $\edge = \strat_m(\vertex, \step, 0) = \strat'(\history_{\vertex, \step}')$; we have $\TP(\history_{\vertex, \step}'\edge) \le \TP(\history\edge)$.
    As $\strat'$ is a winning strategy for $\TPsupge$, necessarily, $(\too(\history_{\vertex, \step}'\edge), \TP(\history_{\vertex, \step}'\edge)) \in \winningExtended$.
    As $\too(\history_{\vertex, \step}'\edge) = \too(\history\edge)$ and $\TP(\history_{\vertex, \step}'\edge) = \TP(\history_{\vertex, \step}') + \col(\edge) \le \TP(\history) + \col(\edge) = \TP(\history\edge)$, we have $(\too(\history\edge), \TP(\history\edge)) \in \winningExtended$.

    If the bit is $1$ after history $\history$, then $\strat_m$ imitates $\strat_\mathsf{safe}$, so $(\too(\history\edge), \TP(\history\edge)) \in \winningExtended$ by definition of $\strat_\mathsf{safe}$.

    \item Let $\edge$ be a possible edge after $\history$.
    Assume the memory bit after $\history\edge$ is $0$.
    This implies that the bit was also $0$ after $\history$.
    By induction hypothesis, there exists a history $\history'$ consistent with $\strat'$ from $\vertex_0$ such that $\card{\history'} = \card{\history}$, $\too(\history') = \too(\history)$, and $\history' \prefixOrd_{\objective_m} \history$.
    We show that $\history_{\vertex, \step}'\edge$ satisfies the desired property for $\history\edge$; we already have by definition that $\card{\history_{\vertex, \step}'\edge} = \card{\history\edge}$ and $\too(\history_{\vertex, \step}'\edge) = \too(\history\edge)$.
    Moreover, as the bit is $0$, $\edge = \strat_m(\vertex, \step, 0) = \strat'(\history_{\vertex, \step}')$, so $\history_{\vertex, \step}'\edge$ is also consistent with $\strat'$.
    By definition of $\history_{\vertex, \step}'$, we also have $\history_{\vertex, \step}'\prefixOrd_{\objective_m} \history'$, therefore $\history_{\vertex, \step}' \prefixOrd_{\objective_m} \history$.
    Hence, $\history_{\vertex, \step}'\edge \prefixOrd_{\objective_m} \history\edge$.
\end{enumerate}
}

This shows the two inductive properties above.
We still have to show that $\strat_m$ is winning for~$\objective_m$ from $\vertex_0$ and never leaves $\winningExtended$.

We prove that $\strat_m$ wins for $\objective_m$.
Observe that $\strat'$ wins in particular for $\objective_m$ (as $\strat'$ wins for $\TPsupge$).
As $\objective_m$ is open and $\arena$ is finitely branching, this means that all histories consistent with $\strat'$ already satisfy $\objective_m$ after a bounded number of steps (\cref{lem:konig}).
This means that for $\step$ sufficiently large, all histories $\history_{\vertex, \step}'$ used to define $\strat_m$ are such that $\col(\history_{\vertex, \step}')\colours^\omega \subseteq \objective_m$.
In particular, for $\step$ sufficiently large, due to how the bit update from $0$ to $1$ is defined, all histories $\history$ consistent with $\strat_m$ reach a bit value of $1$.
We show that any history $\history$ consistent with $\strat_m$ with bit value $1$ is such that $\col(\history)\colours^\omega \subseteq \objective_m$.
Let $\history$ be such a history, and let $\widetilde{\history}$ be its longest prefix with bit value still at $0$.
If $\col(\widetilde{\history})\colours^\omega \subseteq \objective_m$, we are done.
If $\col(\widetilde{\history})\colours^\omega \not\subseteq \objective_m$, as the bit value is still $0$, there is $\widetilde{\history}'$ consistent with $\strat'$ such that $\card{\widetilde{\history}'} = \card{\widetilde{\history}}$, $\too(\widetilde{\history}') = \too(\widetilde{\history})$, and $\widetilde{\history}' \prefixOrd_{\objective_m} \widetilde{\history}$.
We have that $\history'_{\too(\widetilde{\history}), \card{\widetilde{\history}}} \prefixOrd_{\objective_m} \widetilde{\history}'$, so $\history'_{\too(\widetilde{\history}), \card{\widetilde{\history}}} \prefixOrd_{\objective_m} \widetilde{\history}$.
In particular, $\col(\history'_{\too(\widetilde{\history}), \card{\widetilde{\history}}})\colours^\omega \not\subseteq \objective_m$ as well.
By~\eqref{eq:prefixOrdOm}, this implies that $\TP(\history'_{\too(\widetilde{\history}), \card{\widetilde{\history}}}) \le \TP(\widetilde{\history})$.
Let $\edge$ be the next edge in $\history$ after $\widetilde{\history}$.
As the bit after $\widetilde{\history}\edge$ is $1$, this means that $\col(\history'_{\too(\widetilde{\history}), \card{\widetilde{\history}}}\edge)\colours^\omega \subseteq \objective_m$.
As additionally, $\card{\widetilde{\history}\edge} = \card{\history'_{\too(\widetilde{\history}), \card{\widetilde{\history}}}\edge}$ and $\TP(\widetilde{\history}\edge) \ge \TP(\history'_{\too(\widetilde{\history}), \card{\widetilde{\history}}}\edge) \ge -\frac{1}{m}$, we have that $\col(\widetilde{\history}\edge)\colours^\omega \subseteq \objective_m$.
As $\history$ is a continuation of $\widetilde{\history}\edge$, we also have that $\col(\history)\colours^\omega \subseteq \objective_m$.

We have shown that all histories consistent with $\strat_m$ reach bit value $1$ within a bounded number of steps, and that bit value $1$ indicates that any continuation is winning $\objective_m$.
This shows that $\strat_m$ is winning for $\objective_m$.

The fact that $\strat_m$ never leaves $\winningExtended$ from $\vertex_0$ is a direct consequence of~\ref{enum:1}.
\end{proof}

\begin{proof}[Proof of \cref{thm:TPsupge}]
    Using the previous lemma, we now prove \cref{thm:TPsupge}.
    Let $\arena = \arenaFull$ be an arena and $\vertex_0\in\vertices$ be an initial vertex in the winning region of \POne for $\TPsupge$.
    Let the $\memory = (\IN \times \{0, 1\}, (0, 0), \memUpd)$ be the ``step-counter + $1$-bit'' memory structure (we still need to define how and when $\memUpd$ updates the bit).
    We show that there is a winning step-counter + $1$-bit strategy $\strat$ (i.e., $\strat$ is based on $\memory$) from $\vertex_0$.

    We build a winning step-counter + 1-bit strategy $\strat\colon \vertices_1 \times \IN \times \{0, 1\} \to \edges$ inductively.
    Consider the product arena $\arena' = \arena\product \memory$ (in which the bit updates are not fixed yet, and will be fixed inductively).
    We have that $(\vertex_0, (0, 0))$ is in the winning region of $\arena'$.
    The inductive scheme is as follows: for infinitely many $m\in\IN$, for some step bound $k_m\in\IN$, we fix $\strat$ on $\vertices \times \{0,\ldots, k_m - 1\} \times \{0, 1\}$, yielding arena $\arena'_m$.
    We ensure that
    \begin{itemize}
    \item along all histories $\history$ from $\vertex_0$ consistent with $\strat$ of length at most $k_m$, $\winningExtended$ is not left (i.e., $(\too(\history), \TP(\history))\in\winningExtended$), and
    \item the open objective $\objective_{m}$ is already satisfied within $k_m$ steps (i.e., any history of length $k_m$ consistent with $\strat$ only has winning continuations for $\objective_{m}$).
    \end{itemize}

    For the base case, we may assume that we start the induction at $m = -1$ with $k_{-1} = 0$ and $\objective_{-1} = \colours^\omega$.
    We indeed have that from $(\vertex_0, (0, 0))$, the winning region $\winningExtended$ is not left within $k_{-1} = 0$ step and that the open objective $\objective_{-1}$ is already satisfied.

    Assume that we have fixed $\strat$ in $\arena'$ (decisions and bit updates) up to some step bound $k_m$, that $\winningExtended$ is not left within $k_m$ steps, and that $\objective_m$ is already satisfied for all histories of length $k_m$.
    We reset the bit to $0$ after exactly $k_m$ steps: for $\vertex\in\vertices$ and $b\in\{0, 1\}$, we define $\memUpd(\vertex, (k_m - 1, b)) = (k_m, 0)$.
    Fixing $\strat$ up to bound $k_m$ defines an arena $\arena'_m$.

    Since $\winningExtended$ is not left within the first $k_m$ steps, and since no decisions have been fixed after $k_m$ steps, vertex $(\vertex_0, (0, 0))$ is still in the winning region for $\TPsupge$.
    Let $m' = k_m + 1$.
    We apply \cref{lem:TPsupgeHard}: there exists a step-counter + 1-bit strategy $\strat_{m'}$ such that $\strat_{m'}$ is winning for $\objective_{m'}$ from $(\vertex_0, (0, 0))$ and never leaves $\winningExtended$.
    Observe that there are no decisions to make in $\arena'_m$ before $k_m$ steps have elapsed, and that the bit is set to $0$ after exactly $k_m$ steps.
    By \cref{lem:konig}, as $\arena'_m$ is finitely branching and $\objective_{m'}$ is open, there is a bound $k_{m'}$ such that any history consistent with $\strat_{m'}$ of length $k_{m'}$ already satisfies $\objective_{m'}$.
    Since $\objective_{m'}$ cannot be already satisfied before step $m'$, we have $k_{m'} \ge k_m + 1$.
    We define $\arena'_{m'}$ by fixing strategy $\strat'_{m'}$ up to step $k_{m'}$.

    Iterating this procedure defines a step-counter + 1-bit strategy $\strat$ such that, for infinitely many $m\ge 1$, $\strat$ is winning for $\objective_m$.
    As the sequence $\objective_m$ is decreasing ($\objective_1 \supseteq \objective_2 \supseteq \ldots$), we have that $\strat$ is winning for $\objective_m$ for all $m\ge 1$.
    Hence, $\strat$ is winning for $\TPsupge$.
\end{proof}

In general, this proof does not apply uniformly: in the arena of \cref{fig:nonUniformPiTwo}, \POne has no uniformly winning strategy based on a step counter and finite memory from all $s_i$ simultaneously, since \POne needs to exit to $r_0$ arbitrarily far to the right.

\begin{figure}
    \centering
    \begin{tikzpicture}[node distance=2.5cm and 1.5cm]

\node [Max,initial,initial where=left] (s0) at (0,0) {$s_0$};
\node [Max,initial,initial where=left,right =0.8cm of s0] (s1) {$s_1$};
\node [Max,initial,initial where=left,right =0.8cm of s1] (s2) {$s_2$};
\node [right =0.8cm of s2] (s23) {$\cdots$};
\node [Max,initial,initial where=left,right =0.8cm of s23] (s3) {$s_i$};
\node [right =0.8cm of s3] (s34) {$\cdots$};

\node [Max] (t0) at (0,-1.5) {$t_0$};
\node [Max,right =0.8cm of t0] (t1) {$t_1$};
\node [Max,right =0.8cm of t1] (t2) {$t_2$};
\node [right =0.8cm of t2] (t23) {$\cdots$};
\node [Max,right =0.8cm of t23] (t3) {$t_i$};
\node [right =0.8cm of t3] (t34) {$\cdots$};

\node [Max] (r0) at (0,-3) {$r_0$};

\def\dist{-.9}
\draw (s0) to node[label,swap] {$-1$} +(0,\dist) edge (t0);
\draw (s1) to node[label,swap] {$-2$} +(0,\dist) ++(0,\dist) -> (0,\dist);
 \draw (s2) to node[label,swap] {$-3$} +(0,\dist) ++(0,\dist) -> (0,\dist);
\draw[dotted] (s3) to node[label,swap] {$-i$} +(0,\dist) ++(0,\dist) -> (0,\dist);

\draw (t0)  edge[] node[label] {$0$} (t1);
\draw (t1)  edge[] node[label] {$0$} (t2);
\draw (t2)  edge[] node[label] {$0$} (t23);
\draw (t23) edge[] node[label] {$0$} (t3);
\draw[dotted] (t3) edge node[label] {$0$} (t34);

\draw (t0) to node[label,swap] {$0$} +(0,\dist) edge (r0);
\draw (t1) to node[label,swap] {$1$} +(0,\dist) ++(0,\dist) -> (0,-1.5+\dist);
 \draw (t2) to node[label,swap] {$2$} +(0,\dist) ++(0,\dist) -> (0,-1.5+\dist);
\draw[dotted] (t3) to node[label,swap] {$i$} +(0,\dist) ++(0,\dist) -> (0,-1.5+\dist);

\draw (r0) edge[loop right] node[label] {$0$} (r0);
\end{tikzpicture}
    \caption{Arena in which \POne has no \emph{uniformly} winning step-counter + $1$-bit strategy.}
    \label{fig:nonUniformPiTwo}
\end{figure}
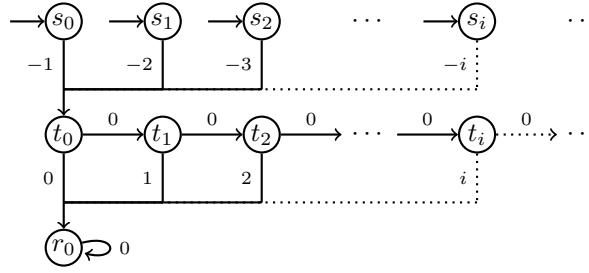

\end{document}